\documentclass[11pt]{article}
\usepackage{url, xcolor}
\usepackage{enumitem}
\usepackage{lscape}
\usepackage{multirow}
\usepackage{pbox}
\usepackage[top=1in, bottom=1in, left=1in, right=1in]{geometry}
\usepackage{graphicx}
\usepackage{subcaption}
\usepackage{epstopdf}
\usepackage[mathscr]{euscript}
\usepackage{amssymb} 
\usepackage{amsmath, comment, setspace,upgreek, natbib, comment, amsthm, tabularx}
\usepackage{algorithm}
\usepackage{algpseudocode}
\usepackage{makecell, chngcntr}
\usepackage[titletoc,title]{appendix}

\newcommand {\Sbar}[1]{{\bar{S}^{#1}}}

\algnewcommand\algorithmicforeach{\textbf{for each}}
\algdef{S}[FOR]{ForEach}[1]{\algorithmicforeach\ #1\ \algorithmicdo}

\def \vhat{\widehat{v}}
\def \E{\mathrm{E}}
\def \V{\mathrm{V}}
\def \Cov{\mathrm{Cov}}
\def \cX{\mathcal{X}}

\def \Pr{\mathrm{Pr}}
\def \ba{\mathbf{a}}
\def \bx{\mathbf{x}}
\def \bz{\mathbf{z}}
\def \bhx{\widehat{\bx}}

\def \blambda{\boldsymbol{\lambda}}
\def \bmu {\boldsymbol{\mu}}

\def \S{S_\alpha(\delta)}
\def \bP{\mathbf{P}}

\def \sk {\mathsf{k}}
\def \sY {\mathsf{Y}}

\def \cZ {\mathcal{Z}}
\def \sL {\mathsf{L}}
\def \cI {\mathcal{I}}
\def \bV {\mathbf{V}}
\def \Ybar {\bar{Y}}
\def \bvartheta {\boldsymbol{\vartheta}}

\newtheorem{theorem}{Theorem}
\newtheorem{proposition}{Proposition}
\newtheorem{assumption}{Assumption}
\newtheorem{lemma}{Lemma}

\newcolumntype{Y}{>{\centering\arraybackslash}X}

\begin{document}

\title{Sequential Bayesian Risk Set Inference for \\
Robust Discrete Optimization via Simulation}

\author{Eunhye Song\\ The Pennsylvania State University}

\date{September 2019}

\maketitle

\section*{Abstract}
Optimization via simulation (OvS) procedures that assume the simulation inputs are generated from the real-world distributions are subject to the risk of selecting a suboptimal solution when the distributions are substituted with input models  estimated from finite real-world data---known as \emph{input model risk}.  Focusing on discrete OvS, this paper proposes a new Bayesian framework for analyzing input model risk of implementing an arbitrary solution, $\bhx$, where uncertainty about the input models is captured by a posterior distribution. 
We define the $\alpha$-level risk set of solution $\bhx$ as the set of solutions whose expected performance is better than $\bhx$ by a practically meaningful margin $(>\delta)$ given common input models  with significant probability ($>\alpha$) under the posterior distribution. The user-specified parameters, $\delta$ and $\alpha$, control robustness of the procedure to the desired level as well as guards against unnecessary conservatism. 
An empty risk set implies that there is no practically better solution than $\bhx$ with significant probability even though the real-world input distributions are unknown---a powerful statistical guarantee. For efficient estimation of the risk set, the conditional mean performance of a solution given a set of input distributions is modeled as a Gaussian process (GP) that takes the solution-distributions pair as an input. In particular, our GP model allows both parametric and nonparametric input models.
We propose the \emph{sequential risk set inference procedure} that estimates the risk set  and selects the next solution-distributions pair to  simulate using the posterior GP at each iteration. We show that simulating the pair expected to change the risk set estimate the most in the next iteration is the asymptotic one-step optimal sampling rule that minimizes the number of incorrectly classified solutions, if the procedure runs without stopping.

\section{Introduction}
\label{sec:intro}

Stochastic simulation requires specification of probabilistic input models to generate random variates from. If the simulation model is built to analyze a real-world system, the input models are often estimated based on observations from the system. These input models are subject to finite-sample estimation error, which is then propagated to the simulation output causing so-called \emph{input uncertainty}. Quantifying input uncertainty for a single simulated system has been actively studied in the literature (see \cite{song2014}, \cite{lam2016} for recent reviews). 
This paper concerns a discrete optimization via simulation (DOvS) problem in the presence of input uncertainty. In particular, we consider a problem with a finite number of feasible solutions in $\mathbb{R}^d$ that share the same input models estimated from the common  input data, where the objective function is defined as the expectation of a simulation output. Furthermore, we focus on the case where the input data are independent and identically distributed (i.i.d.) observations from unknown true real-world input distributions. 
We present the following two examples that arise in applications:
\begin{itemize}
\item A vehicle company performs a small focus-group survey to learn the distribution of customer preferences for vehicle features (e.g. color, seat material, etc.) of a new model. Their goal is to decide the feature mix to release to the market that maximizes the expected profit. A demand input model is fitted to the survey data and a market simulator is used to estimate the expected profit for each candidate feature mix scenario. See~\cite{GMVCO} for details.
\item A regional emergency medical service provider considers relocating their ambulance dispatching center. Their goal is to minimize the average response time, the time between receiving an emergency call from a patient and picking up the patient. The organization has emergency call data (time, location, emergency type, etc.) collected for a year. A discrete-event simulator is used to evaluate a number of candidate locations in the region, where the time and location of emergency calls are randomly generated from an input model fitted from the data. See~\cite{wang2019} for details. We revisit this example for an illustration in Section~\ref{sec:expr}.
\end{itemize}
Common to both examples, input models are derived from finite observations and and additional data collection requires significantly longer time than simulation, nevertheless, the decisions have to be made relying on the available data.  A traditional DOvS approach is designed to select the best solution conditional on the input models being the correct real-world input distributions. A concern for implementing such a conditional optimum is the sensitivity of its relative performance against other solutions to the input models. The conditional optimum may perform poorly under the real-world distributions, if the estimated input models are significantly different from them; such risk is referred to as \emph{input model risk}~\citep{song2015}. 

DOvS under input model risk has been actively studied in recent years, where the focus is typically on optimization---either to find the optimal solution under the unknown real-world distributions or to find a solution to a modified problem that is robust to input model risk---as reviewed in detail in Section~\ref{sec:lit}. This turns out to be a difficult problem; we may not provide a desired probability guarantee of selecting the true optimum without making a strong assumption on relative performances of the solutions or choose a robust optimum that is too conservative under the real-world distributions~\citep{song2019}. Moreover, the former approach often only guarantees the selected solution to be the true optimum asymptotically as the real-world sample size increases to infinity; there is no finite sample guarantee that it is any better than a conditional optimum. 
In many applications, however, practitioners are often okay with selecting a suboptimal solution as long as its optimality gap is within $\delta$, where $\delta$ is the largest optimality gap that can be practically ignored. This is referred to as good selection and has been studied in the ranking and selection (R\&S) literature when the correct distributions are known~\citep{eckman2018}, but not when there is input uncertainty. 

This paper views the DOvS problem from a different angle by focusing on \emph{inference} on the relative performances of the solutions under input uncertainty. In particular, we are interested in the setting when the analyst has chosen a solution, $\widehat{\bx}$, based on some criterion.
\textit{Our goal is to provide a Bayesian framework to measure input model risk of implementing $\widehat{\bx}$ 
that accounts for the largest acceptable optimality gap, $\delta$, given uncertainty about the real-world input distributions.} Adopting Bayesian input modeling, a prior on each real-world input distribution is assumed  and updated to the posterior distribution conditional on the observed data.  
As a measure of risk, we propose the concept of the \emph{risk set}, which contains the solutions that perform practically better ($>\delta$) than  $\widehat{\bx}$ under common input models with significant probability ($>\alpha$). Here, the performance of each solution is measured by the conditional mean of its simulation output given input models sampled from their posterior. The comparison is made conditional on the common input models  and the significance level, $\alpha$, is with respect to the posterior on the input models.


To illustrate implications of the risk set, suppose in the second example above, a parametric input model is assumed for the arrival process of the emergency calls. A prior distribution is assumed for the parameter and its posterior is derived conditional on the data. The analyst found the conditional optimal location ($\widehat{\bx}$) using the maximum a posteriori (MAP) estimate of the posterior to specify the arrival process and is happy to adopt this location if its expected response time is no worse than any other location by $\delta=5$ minutes under the real-world arrival process. Suppose the risk set is constructed at $\alpha=0.1$. Then, each location in the set has the expected response time shorter than $\bhx$ by greater than $5$ minutes under more than $10\%$ of the arrival processes generated by the posterior distribution. If the risk set is empty, it implies that any other candidate location is practically no better than $\bhx$ with $90\%$ confidence: a strong statistical guarantee 
that $\bhx$ is a low-risk solution even in the presence of input uncertainty. 

A strength of the risk set is that it can be used to test robustness of any feasible solution to input model risk, which makes it very practical. For instance, the analyst may prefer a solution for some business reasons that could not be factored into the simulation model and the risk set inference can still provide a valuable insight. 
Also, the risk set guards against unnecessary conservatism by incorporating $\delta$. In the example above, if $\delta=0$, then the risk set includes any solution that is slightly better than $\bhx$ with probability $>\alpha$; if the set size is large, we may conclude that $\bhx$ is high risk. Thus, the same $\bhx$ may be deemed low risk or high risk depending on $\delta$. The user-specified significance level, $\alpha$, also affects the risk set---the smaller $\alpha$ is, the more sensitive the user is to input model risk.  
We show that  the risk set is non-increasing in $\alpha$ as well as $\delta$ in Section~3. 


Estimating the risk set is computationally challenging as it has a nested simulation structure; we first need to estimate the conditional means of all solutions  given a realization from the posterior of input models and then analyze the distribution of the conditional means given the posterior. A naive Monte Carlo (MC) simulation approach is to sample multiple sets of input distributions from the posterior, run an equal number of replications at each solution-distributions pair, and use the sample conditional means to estimate the risk set. However, this is inefficient as some solutions may be far worse than $\bhx$ with a large probability so that it may be  correctly excluded from the risk set with smaller simulation effort, whereas others may require more to be correctly classified.

To allocate simulation effort more efficiently, we design a Bayesian sequential sampling procedure that selects a solution-distributions pair to simulate at each iteration based on a sampling criterion. Specifically, we assume a Gaussian process (GP) prior on the conditional mean surface that takes feasible solution $\bx$ and a collection of input distributions sampled from the posterior as inputs. At each iteration, the GP is updated to its posterior conditional on the cumulative simulation results. Our risk set estimator replaces the sample conditional means in the abovementioned naive MC approach with the GP. 

We propose a sequential sampling rule  to reduce the estimation error of the risk set by \emph{minimizing the expected number of incorrectly classified solutions (to be or not to be in the risk set) in the next iteration}.  
Although the exact loss function depends on the true risk set, we derive a lower bound that only depends on the current estimate of the risk set and show it becomes asymptotically tight as the procedure runs without stopping. We also provide an approximation to the expected lower bound that can be computed cheaply up to a numerical precision without having to estimate.

Our sampling criterion resembles value-of-information-based sampling rules in Bayesian ranking and selection (R\&S) literature \citep{chickinoue2001,frazierKG}. However, our goal is not to select a single solution, but to identify a set of solutions. 
A more closely related literature is on sequential Bayesian experiment design for set estimation as reviewed in Section~\ref{sec:lit}.

Another contribution of this paper is a GP metamodeling framework that can take both parametric and nonparametric input distributions as inputs when the latter is defined by a probability simplex on a finite support (e.g. empirical distribution function). A GP metamodel is a popular tool for developing a DOvS algorithms under input uncertainty \citep{pearce2017,lakshmanan2017} as well as input uncertainty quantification of a single simulated system \citep{xie2014}. All these work consider parametric input models only, which can be easily extended to include nonparametric models using our framework. 

The remainder of the paper is organized as follows. We review related literature in Section~\ref{sec:lit} and mathematically define the risk set in Section~\ref{sec:riskset}, elaborate on Bayesian input modeling and introduce our choice of GP prior in Section~\ref{sec:model}, and discuss the risk set estimator and the sequential sampling criterion in Section~\ref{sec:inference}. The sequential risk set inference procedure and its asymptotic performance analysis are presented in Section~\ref{sec:procedure} followed by its empirical demonstration in Section~\ref{sec:expr}. 

\section{Literature Review}
\label{sec:lit}

Approaches in the literature to solve DOvS problems under input model risk can be divided into several categories according to their problem formulations.  The first is to find the optimal solution, $\bx^c$, under the unknown correct real-world distributions ($c$ for correct). \cite{corlu2013} propose a subset selection procedure that accounts for both simulation error and input uncertainty, where the subset includes $\bx^c$ with probability $\geq 1-\alpha$ under the assumption that the performance of $\bx^c$ given the real-world distributions is at least $\delta$ better than the other solutions' (preference-zone assumption). Under the same assumption, \cite{song2015} incorporate input uncertainty into a R\&S procedure such that probability of correct selection of $\bx^c$ is $\geq 1-\alpha$.   
However, because the real-world distributions are unknown, in both approaches, there is a strictly positive lower bound on $\delta$, which increases when input uncertainty increases. If the user chooses $\delta$ smaller than the lower bound, than the procedures cannot provide the desired probability guarantee.  

\cite{song2019} develop a multiple comparisons with the best procedure free of the PZ assumption, which returns a set that includes $\bx^c$ with high probability. Similar to our setting, they focus on the case all solutions share the same input models estimated from common data and explicitly model the joint effect of input uncertainty to the simulation outputs of the solutions, which they refer to as  common-input-data (CID) effects.  However, their procedure tends to return a large set if there are many feasible solutions, because it  1) does not account for $\delta$ and 2) is designed to meet the target family-wise error rate. 
Our approach does not suffer from these issues because 1) we can exclude solutions whose performances are practically indistinguishable from $\bhx$ using $\delta$ and 2) each solution's classification (whether or not to include in the risk set) only depends on the comparison between the solution and $\bhx$.

Another popular formulation is to find solution ${\bx}(\mathcal{U})$ robust to input model risk given uncertainty set $\mathcal{U}$ of candidate input distributions  postulated from real-world data. Specifically, this approach first finds for each $\bx$ the input distribution that produces the worst-case performance among $\mathcal{U}$, then selects the solution with the best worst-case performance as ${\bx}(\mathcal{U})$.
The same formulation has been studied extensively in the distributionally robust optimization literature (see \cite{delage2010,bental2013} for example). 
In the DOvS context, \cite{fan2013} propose R\&S procedures that guarantee probability of correctly selecting $\bx(\mathcal{U})$ when $\mathcal{U}$ consists of a finite number of candidate distributions. Under the same setting, \cite{gao2017} construct an optimal computing budget allocation (OCBA) scheme to maximize the asymptotic probability of correctly selecting $\bx(\mathcal{U})$ as the simulation effort increases. \cite{lakshmanan2017} use a GP metamodel to estimate the parameters for the OCBA scheme to find $\bx(\mathcal{U})$, where $\mathcal{U}$ is constructed by bootstrapping the real-world data.
A benefit of this formulation is that once $\mathcal{U}$ is decided, finding $\bx(\mathcal{U})$ is only subject to simulation error (not input uncertainty) making it easier to provide a statistical guarantee than finding $\bx^c$. On the other hand, $\bx(\mathcal{U})$ tends to be conservative as it considers the worst-case input distributions for each solution instead of evaluating all solutions under common input distributions. 
Consequently, $\bx(\mathcal{U})$ may be significantly suboptimal under the real-world input distributions.   

Some procedures  aim to find the solution with the best average simulation output, $\bar{\bx}$, where the average is taken with respect to both input uncertainty and simulation error. This objective is said to be risk-neutral to input model risk. \cite{corlu2015} design a subset selection procedure that returns a set containing  $\bar{\bx}$ with probability $\geq 1-\alpha.$
\cite{pearce2017} and \cite{wangng2018} propose GP-based Bayesian optimization (BO) procedures that find $\bar{\bx}$. BO is a popular inference-based optimization procedure when the objective is to find the global optimum of a expensive-to-evaluate function (see \cite{frazierBO} for a recent survey on BO) and can be naturally extended in this context. 
However, such a risk-neutral objective fails to provide any warning to the user when there is large input uncertainty. \cite{wu2018} propose using risk measures instead of the average such as a value-at-risk (VaR) or conditional value-at-risk. 


In terms of methodology, our procedure is closely related to Bayesian sequential design of experiments for set estimation, which has been studied in several different contexts. \cite{picheny2010} consider estimating a level set of a function $g$ of $\mathbf{X}$, e.g., $\{\mathbf{X}:g(\mathbf{X})=\gamma\}$, by modeling $g(\mathbf{X})$ as a GP. 
Their sampling rule targets minimizing the integrated mean squared error (IMSE) in a small neighborhood of the level set, $\{\mathbf{X}: g(\mathbf{X})\in[\gamma\pm\delta]\}$.
Applied to reliability engineering, \cite{Bect2012} discuss estimating the probability of the failure event, $\{\mathbf{X}:g(\mathbf{X})\geq \gamma\}$, given random vector $\mathbf{X}$, which can be converted to estimating a superlevel set of support points of $\mathbf{X}$. 
They also model $g(\mathbf{X})$ as a GP and derive sampling rules to minimize the expected value of the quadratic loss of the estimated failure probability. 
The same problem is considered by \cite{zanette2019}, however, their goal is to provide a robust estimate of the superlevel set by returning the largest set that contains the true set with a target probability.
In the context of DOvS, \cite{xiefrazier2013} study multiple comparisons with known standard value $\gamma$, where the goal is to identify a set of solutions whose expected simulation outputs are greater than $\gamma$. Taking the Bayesian view, they explicitly derive the optimal sampling policies for some special cases. 
What distinguishes our problem from those considered in these work is its nested structure. Our problem can be seen as first finding a superlevel set of input distributions for each solution, then finding a superlevel set of the solutions. Thus, a new sampling criterion is needed to solve it sequentially.


\section{Risk set formulation}
\label{sec:riskset}
In this section, we formally define the risk set. We denote the finite feasible solution space in $\mathbb{R}^d$ by $\cX$. Let  $\bP^c$ represent the true real-world joint input distributions from which we collect observations, $\bz$. Taking a Bayesian approach,  we model $\bP^c$ with $\bP$ by choosing the prior on $\bP$ and updating its posterior $\pi(\bP|\bz)$ conditional on $\bz$. We remain abstract about input modeling in this section and defer its detailed discussion to Section~\ref{subsec:inputmodeling}. 
The simulation output of solution $\bx\in\mathcal{X}$ when the inputs are generated from $\bP$ can be written as 
$Y(\bx;\bP)= \E[Y(\bx;\bP)|\bP] + \varepsilon(\bx;\bP),$
where $\varepsilon(\bx;\bP)$ is the simulation error with mean $0$ and variance $v(\bx,\bP)<\infty$ conditional on $\bP$. We refer to $\E[Y(\bx;\bP)|\bP]$ as the \textit{conditional mean surface}, which maps $(\bx,\bP)$ to $\mathbb{R}$. 

Suppose $\bhx\in\cX$ is the solution considered for implementation. 
Given $\bhx$, we define the $\alpha$-level risk set given the user-chosen error rate $0<\alpha<1$ and the largest optimality gap $\delta\geq 0$ that can be practically ignored as
\begin{equation}
\label{eq:risk.set}
S_\alpha(\delta) \equiv \left\{\bx\in\cX\left|\Pr\{ \E[Y(\bhx;\bP)|\bP]-\E[Y(\bx;\bP)|\bP]>\delta|\bz\}> \alpha \right.\right\},
\end{equation}
where the probability is taken with respect to $\pi(\bP|\bz)$.  
If $S_\alpha(\delta)$ is empty, then 
$
\Pr\{ \E[Y(\bhx;\bP)|\bP]-\E[Y(\bx^c;\bP)|\bP]>\delta|\bz\}\leq \alpha
$
implying that the probability $\bx^c$ is practically significantly better than $\bhx$ is within the acceptable error rate, $\alpha$. 
Note it is possible for $S_\alpha(\delta)$ not to include $\bx^c$ even if $\E[Y(\bhx;\bP)|\bP]<\E[Y(\bx^c;\bP)|\bP]$ with probability $1$. 
In fact, identifying $\bx^c$ is not our goal, nor we make inference on the expected simulation outputs of the solutions under $\bP^c$; \emph{the goal is to find the solutions that pose risk of outperforming $\bhx$ by a significant margin given uncertainty about $\bP$ described by its posterior distribution}. 
Below, we illustrate an example of a risk set using a simple DOvS problem. 

Consider an M/M/$1$/$k$ queueing system. The objective of this problem is to decide the system capacity, i.e., $\bx = k$, that minimizes the expected cost of operation: 
\begin{equation}
\label{eq:mm1k.cost}
\E[\mbox{net cost per customer}] = c\E[\mbox{waiting time}] - r(1-\Pr\{\mbox{balking}\}),
\end{equation}
where $c$ is the cost per unit waiting time per customer and $r$ is the revenue per served customer. We assume $\cX=\{\bx|1 \leq \bx \leq 50\}$. 
In this problem, $\bP$ is the joint distribution of interarrival and service times. Suppose real-world interarrival and service times are exponentially distributed with means $1$ and $1.1$, respectively, and we have $100$ i.i.d.\ observations from each. From the steady-state analysis of the $M/M/1/k$ queue, one can obtain the analytical expression for~(\ref{eq:mm1k.cost}) (see Appendix~A).  Figure~\ref{fig:cost_surface} plots the expected cost~(\ref{eq:mm1k.cost}) as a function of $\bx$ and the mean service time when the arrival rate $= 1$, $c=1$, and $r=200$. The optimal solutions at selected mean service times are marked with white solid circles, which shows their clear dependence.
Given the true mean service time, $1.1$, the correct optimal solution is $\bx^c = 14$, however, Figure~\ref{fig:cost_surface} shows that if we simply assume the mean service time estimated from the data as the truth, we may choose a different solution as the best. Nevertheless, we argue that the solution is still acceptable, if its risk set is empty given user-specified $\alpha$ and $\delta$. 

\begin{figure}
\centering
\begin{minipage}{0.45\textwidth}
\centering
\includegraphics[scale=0.53]{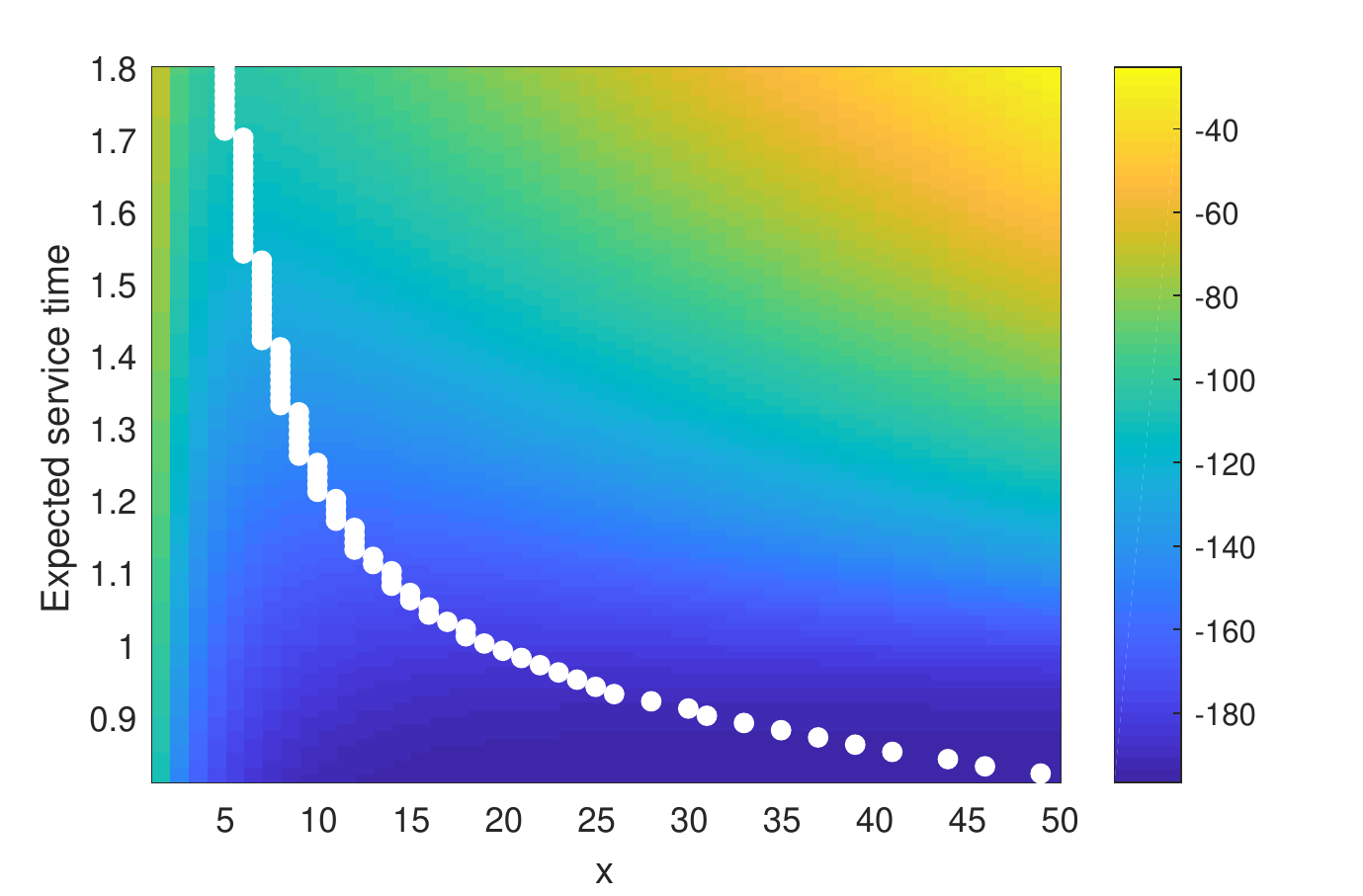}
\caption{A plot of the expected cost of the M/M/$1$/$k$ example. Cost-minimizing $\bx = k$ for each expected service time is marked with a white circle.}
\label{fig:cost_surface}
\end{minipage} ~
\begin{minipage}{0.45\textwidth}
\centering
\captionsetup{type=table}
\caption{Risk sets for different values of $\alpha$ when $\delta = 1$ and $\bhx = 17$ given a sample of $100$ i.i.d.\ observations from Exp($1$) and Exp($1.1$).  }
\begin{tabular}{l|l}
$\alpha$ & \multicolumn{1}{c}{$S_\alpha(\delta)$} \\ \hline
$0.05$ & $\{5, 6, 7, 8, 9, 10, 11, 12, 13, 14, 15, 16\}$\\ \hline
$0.1$ & $\{7, 8, 9, 10, 11, 12, 13, 14, 15\}$\\\hline
$0.15$ & $\{8, 9, 10, 11, 12, 13, 14, 15\}$\\\hline
$0.2$ & $\{10, 11, 12, 13, 14, 15\}$\\\hline
$0.25$ & $\{12, 13\}$\\\hline
$0.3$ & $\emptyset$\\
\end{tabular}
\label{tab:riskset.demo}
\end{minipage}
\end{figure}

For demonstration, we sampled $100$ observations from each real-world input distribution,  modeled $\bP$ assuming the nonparametric prior described in Section~\ref{subsec:inputmodeling} on $\bP$ and computed its posterior distribution conditional on the observations. Given the MAP of the posterior, $\bx = 17$ minimizes the expected cost.  
Table~\ref{tab:riskset.demo} shows the estimated risk sets for different values of $\alpha$ when $\delta = 1$ and $\bhx = 17$, which are obtained by sampling $1{,}000$ $\bP$s from its posterior distribution and estimating the probability in~(\ref{eq:risk.set}) for each $\bx$.
The majority of the solutions are excluded from the risk sets for all $\alpha$ values in Table~\ref{tab:riskset.demo}. Also, none of the risk sets include $\bx = 18$, because its performance does not differ significantly from that of $\bhx$. Figure~\ref{subfig:x18} displays the histogram of $\E[Y(17;\bP)|\bP]-\E[Y(18;\bP)|\bP]$ obtained from the sampled $\bP$s. All observations are less than $\delta = 1$ explaining why $\bx=18$ is excluded from all estimated risk sets in Table~\ref{tab:riskset.demo}. 
On the other hand, $\bx = 13$ is included in $\S$ for all $\alpha$'s except for $\alpha = 0.3$ implying that it performs so differently from $\bhx$ that they outperform $\bhx$ with a significant margin ($>\delta$) under some $\bP$ (see Figure~\ref{subfig:x13}). Some other solutions such as $\bx=9$ are included in the risk sets only for smaller $\alpha$ values (Figure~\ref{subfig:x9}). Also notice that $\bx^c = 14$ is not always included in $\S$ as discussed earlier.

\begin{figure} [tb]
\begin{subfigure}{.33\textwidth}
\begin{center}
\includegraphics[scale=0.58]{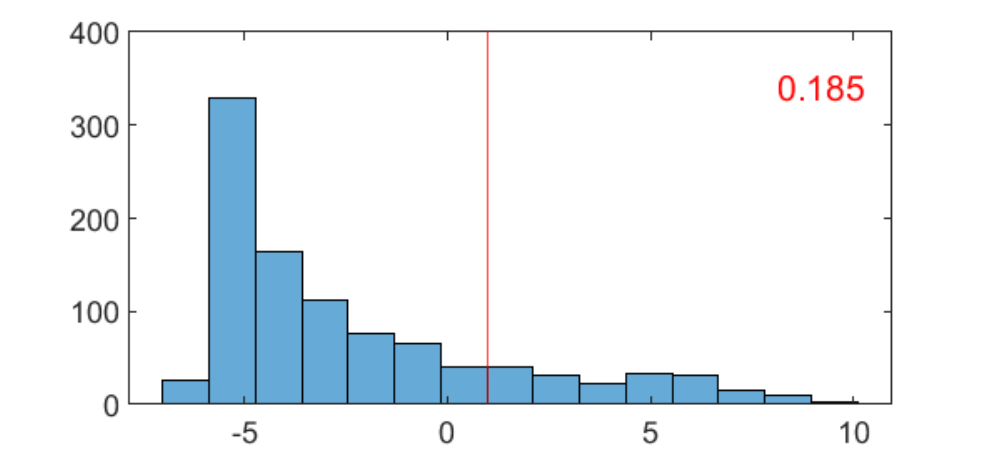}
\caption{$\bx = 9$}
\label{subfig:x9}
\end{center}
\end{subfigure}%
\begin{subfigure}{.33\textwidth}
\begin{center}
\includegraphics[scale=0.58]{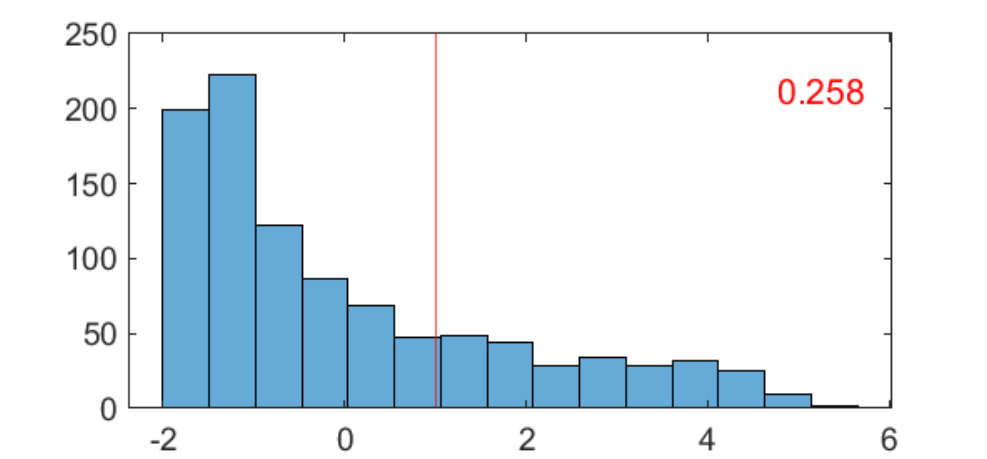}
\caption{$\bx = 13$}
\label{subfig:x13}
\end{center}
\end{subfigure}%
\begin{subfigure}{.33\textwidth}
\begin{center}
\includegraphics[scale=0.58]{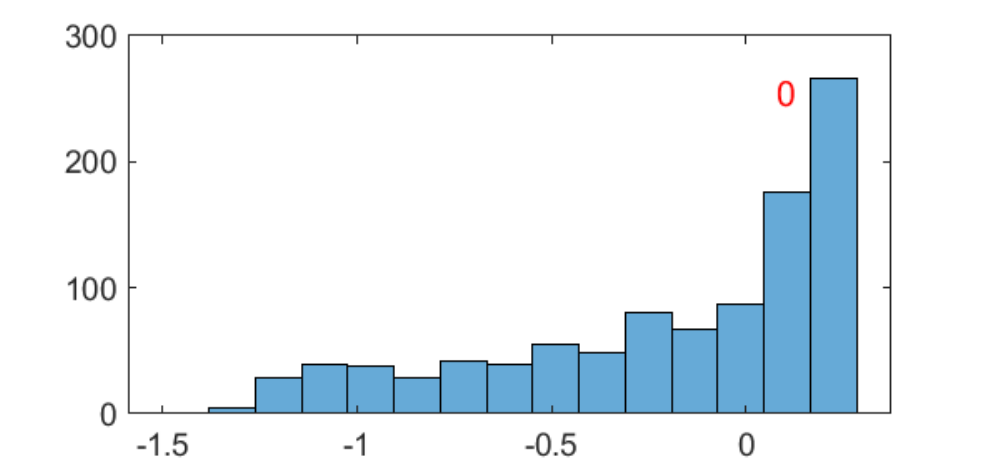}
\caption{$\bx = 18$}
\label{subfig:x18}
\end{center}
\end{subfigure}
\caption{The histograms of $\E[Y(\bhx;\bP)|\bP]-\E[Y(\bx;\bP)|\bP]$ for $\bx = 9, 13,$ and $18$ obtained by sampling $1{,}000$ $\bP$s from its posterior distribution. The red line represents $\E[Y(\bhx;\bP)|\bP]-\E[Y(\bx;\bP)|\bP] = \delta$. The estimate of $\Pr\{ \E[Y(\bhx;\bP)|\bP]-\E[Y(\bx^c;\bP)|\bP]>\delta|\bz\}$ is presented at the upper right corner of each plot.}
\end{figure}

\sloppy
Another interesting observation from Table~\ref{tab:riskset.demo} is that
 $S_{\alpha}(\delta)$ is non-increasing in $\alpha$. 
That is, $S_{\alpha_1}(\delta) \supseteq S_{\alpha_2}(\delta)$ if $\alpha_1 \leq \alpha_2$, which can be shown directly from Definition~(\ref{eq:risk.set}). Smaller $\alpha$ means that the decision maker is more cautious about input model risk and wants to detect solutions with smaller chances of being practically better than $\bhx$ resulting a larger risk set. In the example above, the risk set becomes empty for $\alpha \geq0.3$.
Moreover, $S_\alpha(\delta)$ is non-increasing in $\delta$, because for each $\bx\in\mathcal{X}$ and any $\delta1\leq\delta_2$, $\Pr\{ \E[Y(\bhx;\bP)|\bP]-\E[Y(\bx;\bP)|\bP]>\delta_1|\bz\}\geq\Pr\{ \E[Y(\bhx;\bP)|\bP]-\E[Y(\bx;\bP)|\bP]>\delta_2|\bz\}$. A decision maker with smaller $\delta$ cares about smaller optimality gap, and thus  ends up with a bigger---non-smaller to be precise---risk set.

The following expression for $S_\alpha(\delta)$ is equivalent to~(\ref{eq:risk.set}):
\begin{equation}
\label{eq:risk.set2}
S_\alpha(\delta) = \{\bx\in\cX|q_{1-\alpha}(\E[Y(\bhx;\bP)|\bP]-\E[Y(\bx;\bP)|\bP])>\delta\},
\end{equation}
where $q_{1-\alpha}(\cdot)$ is the $(1-\alpha)$-level quantile function with respect to $\pi(\bP|\bz)$.  By definition, $q_{1-\alpha}$ can be replaced by $\mathrm{VaR}_\alpha(\cdot)$, the $\alpha$-level value-at-risk. Therefore,~(\ref{eq:risk.set2}) can be interpreted as applying a robust measure controlled by $\alpha$ to decide whether $\bx$ may perform better than $\bhx$.
Reformulation~(\ref{eq:risk.set2}) shows that we can devise a class of the risk sets by replacing VaR$_\alpha(\cdot)$ with other risk measures. For instance, if conditional value at risk (CVaR$_\alpha(\cdot)$) is used instead of VaR$_\alpha(\cdot)$, then $\S$ includes $\bx$ such that whose conditional $\alpha$-tail average of ${\E[Y(\bhx;\bP)|\bP]-\E[Y(\bx;\bP)|\bP]}$ is greater than $\delta$. 
We focus on~(\ref{eq:risk.set}) leaving the possibility of exploring other risk measures in future work.
 
Similar to~\cite{song2019}, our formulation of $\S$ facilitates exploiting the CID effects when there is positive correlation between $\E[Y(\bhx;\bP)|\bP]$ and $\E[Y(\bx;\bP)|\bP]$ for $\bx\neq \bhx$. Loosely speaking, such positive correlation makes it easier to decide whether $\bx$ belongs in $\S$ or not. 
To see this, let $\mu_\Delta$ and $\sigma^2_\Delta$ be mean and variance of $\Delta\equiv \E[Y(\bhx;\bP)|\bP] - \E[Y(\bx;\bP)|\bP]$, respectively. Recall that $\bx\in\S$, if $\Pr\{\Delta> \delta\} \geq \alpha$. 
When $\mu_\Delta < \delta$, 
\begin{equation}
\Pr\{\Delta> \delta\}=\Pr\{\Delta-\mu_\Delta>\delta-\mu_\Delta\} \leq \Pr\{|\Delta-\mu_\Delta|\geq \delta-\mu_\Delta\} \leq \sigma_\Delta^2/(\delta-\mu_\Delta)^2, \label{eq:CIDeffects.prob}
\end{equation}
where the last inequality follows from Chebyshev's inequality. The smaller $\sigma^2_\Delta$ is, the smaller the upper bound of~(\ref{eq:CIDeffects.prob}) is. 
Positive correlation between $\E[Y(\bx;\bP)|\bP]$ and $\E[Y(\bhx;\bP)|\bP]$ reduces $\sigma^2_\Delta$ and if the upper bound of~(\ref{eq:CIDeffects.prob}) is smaller than $\alpha$, $\bx$ is excluded from the $\alpha$-level risk set.  
When $\mu_\Delta > \delta$, 
\begin{equation}
\label{eq:CIDeffects.prob2}
\Pr\{\Delta> \delta\} = 1-\Pr\{\Delta-\mu_\Delta \leq \delta - \mu_\Delta\}  \geq 1-\Pr\{|\Delta-\mu_\Delta|\geq \mu_\Delta-\delta\} \geq
1-\sigma_\Delta^2/(\mu_\Delta-\delta)^2.
\end{equation}
Similarly, the last inequality of~(\ref{eq:CIDeffects.prob2}) follows from Chebyshev's inequaltiy. In this case, the smaller $\sigma^2_\Delta$ is, the larger the lower bound of~(\ref{eq:CIDeffects.prob2}) is. 
Thus, when $\E[Y(\bx;\bP)|\bP]$ and $\E[Y(\bhx;\bP)|\bP]$ are positively correlated, we have larger lower bound for~(\ref{eq:CIDeffects.prob2}) and if the bound is larger than $\alpha$, then $\bx$ is included in the $\alpha$-level risk set.

\vspace{11pt}
\noindent
\textbf{Remark:} As mentioned in Section~\ref{sec:lit}, the risk set can be interpretted as a (strict) superlevel set that includes the solutions satisfying the inequality constraint ($>\alpha$) on the probability in~(\ref{eq:risk.set}). Moreover, for each $\bx$, estimating the probability in~(\ref{eq:risk.set}) can be interpreted as estimating a level set of $\bP$ that satisfies the condition, $\E[Y(\bhx;\bP)|\bP]-\E[Y(\bx;\bP)|\bP]>\delta$. 

\section{Bayesian modeling}
\label{sec:model}
In this section, we discuss Bayesian modeling of $\bP$ as well as the choice of a prior to model the conditional mean surface as a GP that takes $(\bx,\bP)$ as inputs. Estimating the risk set has two major challenges. First, unlike the simple $M/M/1/k$ example above, for a general DOvS problem the analytical expression for $\E[Y(\bx;\bP)|\bP]$ is unknown and must be estimated via simulation. Secondly, from the posterior distribution of $\bP$ we need to estimate the posterior distribution of $\E[Y(\bhx;\bP)|\bP]-\E[Y(\bx;\bP)|\bP]$ for each $\bx\neq\bhx$.  We discuss how the GP model facilitates efficient risk set estimation by providing posterior inference on the conditional mean surface.

\subsection{Bayesian input modeling}
\label{subsec:inputmodeling}
Suppose there are $L$ real-world input distributions from which we collect $m_1,m_2,\ldots,m_L$ observations. Let $z_1^\ell, z_2^\ell, \ldots, z_{m_\ell}^\ell \in \mathcal{Z}^\ell$ be observations from the $\ell$th input process whose unknown true probability measure is $P_\ell^{c}$, where $\mathcal{Z}^\ell$ may be a subset of $\mathbb{R}^d$ or even a more general space. \textit{We allow the inputs to be dependent when they are observed together}, e.g.,\ the first input may be a three-dimensional random vector whose elements are correlated. The joint true probability measure is denoted by $\bP^c$, which is simply a product of $P_1^{c}, P_2^{c},\ldots,P_L^{c}$ under the independence assumption, and defined on $\cZ \equiv \cZ^1\times\cZ^2\times\cdots\times\cZ^L$. 

We model unknown $P_\ell^{c}$ with $P_\ell$ by imposing a prior distribution on $P_\ell$. 
If we take a parametric Bayesian view, then this simply comes down to assuming a prior distribution on the parameter vector, $\theta_\ell$, of $P_\ell$. When the distribution family of $P_\ell$ allows a conjugate prior on $\theta_\ell$, then the posterior distribution of $\theta_\ell$ can be obtained analytically. 
Although this is a popular approach in the input uncertainty literature \citep{ng2006,xie2016}, it restricts input modeling flexibility. When a conjugate relationship is not available, one can still sample from the posterior distribution via Markov Chain Monte Carlo, however, it may be computationally expensive.

An alternative approach is to model $P_\ell$ nonparametically by representing $P_\ell$ as a probability simplex on $z_1^\ell, z_2^\ell, \ldots, z_{m_\ell}^\ell$ and captures uncertainty about $P_\ell$ by putting a prior distribution on the probability simplex.
This has three major benefits over the parametric approach: 1) no assumption on the distribution family of $P_\ell$ needs to be made; 2) the posterior update for $P_\ell$ is straightforward; 3)  sampling from the posterior of $P_\ell$ is very cheap. In the following, we introduce the nonparametric Bayesian input modeling in detail.

Let $u_\ell$ denote the number of distinct observations among $z_1^\ell, z_2^\ell, \ldots, z_{m_\ell}^\ell$ and $v_1^\ell, v_2^\ell,\ldots,v_{u_\ell}^\ell$ denote those observations. We further define $\mathbf{w}^\ell=\{w_1^\ell,w_2^\ell,\ldots,w_{u_\ell}^\ell\}$ as the probability simplex assigned to $\{v_1^\ell, v_2^\ell,\ldots,v_{u_\ell}^\ell\}$.
We assume a Dirichlet prior on $\mathbf{w}^\ell$ whose density function is proportional to
\begin{equation}
\label{eq:prior}
I\left\{\sum_{j=1}^{u_\ell}w_j^\ell=1\right\}\Pi_{j=1}^{u_\ell} (w_j^\ell)^{\kappa_j^\ell-1}, 
\end{equation}
where $I(\cdot)$ is an indicator function and $\kappa_j^\ell>0, 1\leq j \leq u_\ell$, are the concentration parameters of the Dirichlet distribution.  
Assuming $P_\ell$ is defined by the probability simplex $\mathbf{w}^\ell$, the posterior distribution of $\mathbf{w}^\ell$ is again Dirichlet with its distribution function proportional to 
\begin{equation}
\label{eq:dirichlet}
I\left\{\sum_{j=1}^{u_\ell}w_j^\ell=1\right\}\Pi_{j=1}^{u_\ell} (w_j^\ell)^{c^\ell_j+\kappa^\ell_j-1},
\end{equation}
where $c^\ell_j$ represents the number of observations equal to $v_j^\ell$. If all $L$ input distributions are modeled nonparametrically, the posterior distribution function of $\bP$ given observed data $\bz = \{z_j^\ell,j=1,2,\ldots,m_\ell, \ell=1,2,\ldots,L\}$, $\pi(\bP|\bz)$, is proportional to 
\begin{equation}
\label{eq:posterior.P}
\Pi_{\ell=1}^L \left(I\left\{\sum_{j=1}^{u_\ell}w_j^\ell=1\right\}\Pi_{j=1}^{u_\ell} (w_j^\ell)^{c^\ell_j+\kappa^\ell_j-1}\right).
\end{equation}
When mixed---parametric and nonparametric---input models are used, we can still write a similar product form for~(\ref{eq:posterior.P}) under the independence assumption. 

\cite{rubin1981} names the abovementioned method \emph{Bayesian bootstrap} as the resulting input model resamples the data points with the weights given by the probability simplex. Bayesian bootstrap is more flexible than frequentist's nonparametric bootstrap approach as it allows $w_j$ to be continuous-valued rather than a multiple of $1/m_\ell$ as in  frequentist's bootstrap, where $\kappa_j^\ell$'s control how concentrated the distribution of the probability simplex is. 
In particular, if $\kappa_j^\ell=1$ for all $j$ and $\ell$ for the prior distribution, then it implies that all probability simplices defined on $\{v_1^\ell, v_2^\ell,\ldots,v_{u_\ell}^\ell\}$ are equally likely (uniform Dirichlet). 
With this choice, the concentration parameters of the posterior distribution~(\ref{eq:dirichlet}) become $c_j^\ell +1$, for $j=1,2,\ldots, u_\ell,\ell=1,2,\ldots,L,$ and the MAP of the posterior distribution is $w_j^\ell = c_j^\ell/m_\ell$ for all $j$, which is equivalent to the frequentist's empirical distribution function of $z_1^\ell, z_2^\ell, \ldots, z_{m_\ell}^\ell$~\citep{owen2001}.  The larger $c_j^\ell$'s are, the more concentrated the distribution of the probability simplex is around its MAP. 

A probability simplex on $\bz^\ell$ that follows~(\ref{eq:dirichlet}) can be generated from $m^\ell + \sum_{j=1}^{u_\ell} \kappa_j^\ell-1$ independent uniform $(0,1)$ random variables. See \cite{rubin1981} for further details.

\subsection{Gaussian process model for the conditional mean surface}
\label{subsec:prior}

A GP model has been extremely popular for representing an expensive-to-evaluate objective function in BO~\citep{frazierBO}. In their context, the goal is to find $\arg\min_{\bx\in\cX}g(\bx)$, where the analytical expression of $g$ is unknown, but $g$ can be evaluated at $\bx$ either exactly or with an additional stochastic error. 
In the context of OvS under input uncertainty, \cite{pearce2017},~\cite{lakshmanan2017}, and \cite{wangng2018} adopt a GP to model $\E[Y(\bx;\bP)|\bP]$ when $\bP$ is characterized by its parameter vector, $\boldsymbol{\theta} = \{\theta_1,\theta_2,\ldots,\theta_L\}$. In this case, both $\bx$ and $\boldsymbol{\theta}$ are defined in $\mathbb{R}^d$ letting them use the same GP model typically used in BO. 
However, when some input distributions are modeled nonparametrically as described in Section~\ref{subsec:inputmodeling}, we need a GP that takes $\bx$ and a probability simplex as inputs. 
Below, we present our choice of GP prior that facilitates using the nonparametric input models and discuss its posterior update. 

For simplicity of presentation, we assume all $L$ input distributions are modeled nonparametrically with the Dirichlet prior/posterior in Section~\ref{subsec:inputmodeling}. However, our GP model can be extended easily to the case of mixed input models; see remarks below.
We model $\E[Y(\bx;\bP)|\bP]$ as a realization of GP $\eta(\bx,\bP)$ with mean function $\mu(\bx,\bP)$ and covariance kernel $\mathsf{k}(\bx,\bP;\bx^\prime,\bP^\prime)$:
\begin{equation}
\label{eq:GP}
\eta(\bx,\bP) \sim \mathcal{GP}(\mu(\bx,\bP),\sk(\bx,\bP;\bx^\prime,\bP^\prime)) \mbox{ for } \bx,\bx^\prime \in\cX, \bP, \bP^\prime \in \mathcal{M}_+^1(\mathcal{Z}),
\end{equation}
where $\mathcal{M}_+^1(\mathcal{Z})$ represent a set of probability simplices on $\mathcal{Z}$.  
The covariance between $\eta(\bx,\bP)$ and $\eta(\bx^\prime, \bP^\prime)$ is defined by covariance kernel $\sk$, i.e., $\Cov(\eta(\bx,\bP),\eta(\bx^\prime, \bP^\prime))= \sk(\bx,\bP;\bx^\prime,\bP^\prime)$.  
The covariance kernel is a driving force of GP prediction as it enables statistical inference on $\E[Y(\bx;\bP)|\bP]$ at $(\bx,\bP)$ that has not been simulated yet by inducing spatial correlation between $(\bx,\bP)$ and other solution-distributions pairs that have been simulated. 
Any finite number of observations from the GP have a joint Gaussian distribution whose mean and variance-covariance matrix are defined by $\mu$ and the Gram matrix constructed from $\sk$~\citep{rasmussen2005}. 
We choose $\mu(\bx,\bP) = \beta_0\in \mathbb{R}$ (i.e., constant prior mean) and   covariance kernel
$\sk(\bx,\bP;\bx^\prime,\bP^\prime) = \tau^2 \gamma_\cX(\bx,\bx^\prime) \gamma_\mathcal{M}(\bP,\bP^\prime)$,
where $\tau^2\in\mathbb{R}$ is the marginal variance of the GP, and $\gamma_\cX$ and $\gamma_\mathcal{M}$ are correlation kernels defined on $\cX\times\cX$ and $\mathcal{M}_+^1(\mathcal{Z})\times \mathcal{M}_+^1(\mathcal{Z})$, respectively. We also require $\sk(\bx,\bP;\bx^\prime,\bP^\prime) > 0$ for all $(\bx,\bP)$ and $(\bx^\prime, \bP^\prime)$, i.e., all $\eta(\bx,\bP)$  are positively correlated under the GP prior. 
We additionally assume  $\varepsilon(\bx,\bP)\sim N(0,v(\bx,\bP))$ so that the posterior  of (\ref{eq:GP}) is also GP. 

In general, kernel $\gamma(\ba,\ba^\prime)$ defined on $\mathcal{A}\times\mathcal{A}$ is said to be positive definite (pd), if and only if,
$\gamma$ is symmetric and 
$
\sum_{1\leq i\leq n}\sum_{1\leq j\leq n} c_i c_j\gamma(\ba_i,\ba_j) \geq 0
$
for all $n\in \mathbb{N}, \ba_i\in \mathcal{A}, i=1,2,\ldots,n$ and for all $c_i \in \mathbb{R}, i=1,2,\ldots,n$.
Such $\gamma$ generates a positive semi-definite Gram matrix given any finite collection of $\ba$'s in $\mathcal{A}$, thus any legitimate correlation kernel of GP must be positive definite.
There are several candidates for $\gamma_\cX$ defined on $\mathbb{R}^d\times \mathbb{R}^d$ such as squared exponential and Mat\'{e}rn classes~\citep{rasmussen2005}. For the examples in Section~\ref{sec:expr}, the following squared exponential kernel is used:
$\gamma_\cX(\bx, \bx^\prime) = \exp\left\{-\sum_{s=1}^d (x_s-x^\prime_s)^2/\lambda_s\right\}, $
where $\lambda_s>0, s=1,2,\ldots,d$.  Observe that for each coordinate $s$, $\gamma_\cX(\bx, \bx^\prime)$ is a decreasing function of $|x_s-x^\prime_s|$. Loosely speaking, $\eta(\bx,\bP)$ and $\eta(\bx^\prime,\bP)$ have higher correlation when $\bx$ and $\bx^\prime$ are close in $\cX$; a modeling assumption makes sense in many DOvS settings. For instance, in Figure~\ref{fig:cost_surface},   when the capacities are close, the expected costs are also close for a fixed expected service time.
For  $\gamma_\mathcal{M}$, we adopt the following:
\begin{equation}
\label{eq:Pkernel}
\gamma_\mathcal{M}(\bP,\bP^\prime) = \exp\left\{-\sum_{\ell=1}^L D^2(P_\ell,P^{\prime}_\ell)/\vartheta_\ell\right\}, 
\end{equation}
where $D(P_\ell,P^{\prime}_\ell)$ is some measure of closeness between the $\ell$th input models, $P_\ell$ and $P^{\prime}_\ell$, of $\bP$ and $\bP^\prime$, respectively, and $\vartheta_\ell>0$ for $\ell=1,2,\ldots,L$. 
In the space of probability distributions, $f$-divergence is a popular class of measures of closeness. 
However, not all $f$-divergences result in a pd kernel when used as $D(\cdot,\cdot)$ in~(\ref{eq:Pkernel}). For instance, the popular Kullback-Leibler (KL) divergence does not produce pd $\gamma_\mathcal{M}$. 
We introduce the sufficient conditions derived by~\cite{hein2004} for  $D(\cdot,\cdot)$ to produce positive definite $\gamma_\mathcal M$ in Appendix B. Some popular $f$-divergences such as the total variation, squared Hellinger distance, and Jenson-Shannon divergence satisfy the conditions (see Table~B.1 in Appendix B). 

\vspace{11pt}
\noindent
\textbf{Remark}: 
Since (\ref{eq:Pkernel}) is a product of $L$ kernels for $L$ independent inputs, it can be modified to take both parametric and nonparametric $P_\ell$'s. 
For instance, if $P_1$ is parametric and the squared exponential kernel is adopted for $P_1$, then $\gamma_\mathcal{M}(\bP,\bP^\prime) = \exp\left\{-(\theta_1-\theta_1^\prime)^2/\vartheta_1-\sum_{\ell=2}^L D^2(P_\ell,P^{\prime}_\ell)/\vartheta_\ell\right\}$, where $\theta_1$ and $\theta_1^\prime$ are the parameter vectors of $P_1$ and $P_1^\prime$, respectively. 
\vspace{11pt}

\sloppy
The parameters of our GP prior, $\beta_0, \tau^2, \blambda= \{\lambda_1,\lambda_2,\ldots,\lambda_d\}$ and $\bvartheta = \{\vartheta_1,\vartheta_2,\ldots,\vartheta_L\}$, can be estimated via maximum likelihood estimation (MLE) after sampling $n_0$ initial design solution-distributions pairs (design pairs for short) $(\bx_1,\bP_1), (\bx_2,\bP_2), \ldots, (\bx_{n_0},\bP_{n_0})$ and simulating $r_i \geq 2$ times at each $(\bx_i, \bP_i)$ (see Appendix C for details). The GP prior is then updated to its posterior conditional on the simulation results form the $n_0$ design pairs. More generally, suppose $n_t$ distinct $(\bx,\bP)$ pairs are sampled by the $t$th iteration of our sequential algorithm.  Let $\sY_t =\{\Ybar_1,\Ybar_2,\ldots,\Ybar_{n_t}\}$, where $\Ybar_i = \sum_{j=1}^{r_i}Y_j(\bx_i,\bP_i)/r_i$ and $Y_j(\bx_i,\bP_i)$ is the $j$th simulation output. Then, the joint posterior distribution of $\eta(\bx, \bP)$ and $\eta(\bx^\prime,\bP^\prime)$ for any $\bx, \bx^\prime \in\cX$ and $\bP,\bP^\prime\in\mathcal{M}_+^1(\mathcal{Z})$ conditional on $\sY_t$ is 
\small
\begin{align}
\label{eq:GPposterior}
& \left.\left(\begin{matrix}
\eta(\bx,\bP) \\
\eta(\bx^\prime,\bP^\prime)
\end{matrix} 
\right)\right|
\mathsf{Y}_t
\sim 
N\left(
\beta_0\mathbf{1}_2 - 
\left[
\begin{matrix}
\Sigma_t(\bx,\bP)^\top \\
\Sigma_t(\bx^\prime,\bP^\prime)^\top 
\end{matrix}
\right]
(\Sigma_t+{\Sigma}_t^\varepsilon)^{-1}(\mathsf{Y}_t-\beta_0\mathbf{1}_{n_t}), \right.   \\
&\left.
\tau^2 \left[
\begin{matrix}
1 & \gamma_\cX(\bx,\bx^\prime)\gamma_\mathcal{M}(\bP,\bP^\prime) \\
\gamma_\cX(\bx,\bx^\prime)\gamma_\mathcal{M}(\bP,\bP^\prime)  & 1
\end{matrix}
\right]
-
\left[
\begin{matrix}
\Sigma_t(\bx,\bP)^\top \\
\Sigma_t(\bx^\prime,\bP^\prime)^\top 
\end{matrix}
\right]
(\Sigma_t+\Sigma_t^\varepsilon)^{-1}
\left[
\begin{matrix}
\Sigma_t(\bx,\bP) & 
\Sigma_t(\bx^\prime,\bP^\prime)
\end{matrix}
\right]
\right),  \nonumber
\end{align} \normalsize
where $\Sigma_t(\bx,\bP)$ is $n_t \times 1$ covariance vector between $(\bx,\bP)$ and the $n_t$ simulated pairs constructed by kernel $\mathsf{k}$, $\Sigma_t$ is the Gram matrix of $\mathsf{k}$ and $n_t$ simulated pairs, and $\Sigma_t^\varepsilon$ is the variance-covariance matrix of the simulation errors of $\sY_t$. When all $(\bx,\bP)$ pairs are simulated independently, $\Sigma_t^\varepsilon$ is simply a $n_t\times n_t$ diagonal matrix with $v(\bx_i,\bP_i)/r_i$ on the $i$th diagonal. Since $v(\bx_i,\bP_i)$ is unknown, we use the sample variance, $S^2(\bx_i,\bP_i)=\sum_{j=1}^{r_i}(Y_j(\bx_i,\bP_i)-\bar{Y}(\bx_i,\bP_i))^2/(r_i-1)$, as its plug-in estimator.
In the following sections, we denote the posterior mean of GP at $(\bx,\bP)$ in the $t$th iteration by $\mu_t(\bx,\bP)$ for notational convenience.

\section{Risk set inference}
\label{sec:inference}

At the $t$th iteration, $S_\alpha(\delta)$ can be estimated by replacing $\E[Y(\bx;\bP)|\bP]$ in~(\ref{eq:risk.set}) with GP $\eta(\bx,\bP)$:
\begin{equation}
\label{eq:S.est}
\left\{\bx\in\cX\left|\Pr\{\eta(\bhx,\bP)-\eta(\bx,\bP)>\delta|\sY_t, \bz\}> \alpha \right.\right\},
\end{equation}
where the probability is taken with respect to the joint posterior distribution of $\bP$ and the GP conditional on $\sY_t$ and $\bz$. 
The probability in~(\ref{eq:S.est}) can be rewritten as
\begin{align}
&\int_{\bP} \Pr\{\eta(\bhx,\bP)-\eta(\bx,\bP)>\delta|\sY_t,\bP\}\pi(\bP|\bz)d\bP  \nonumber \\
&= \int_{\bP} \Pr\{\eta(\bhx,\bP)-\eta(\bx,\bP)-(\mu_t(\bhx,\bP)-\mu_t(\bx,\bP))>\delta - (\mu_t(\bhx,\bP)-\mu_t(\bx,\bP))|\sY_t,\bP\}\pi(\bP|\bz)d\bP \nonumber \\
&= \int_\bP\Phi\left(\frac{\mu_t(\bhx,\bP)-\mu_t(\bx,\bP)-\delta}{\sigma_t(\bhx,\bx,\bP)}\right)\pi(\bP|\bz)d\bP \label{eq:prob.integration}
\end{align}
where $\Phi(\cdot)$ is the cumulative distribution function (cdf) of the standard normal distribution and $\sigma_t(\bhx,\bx,\bP)$ is the standard deviation of $\eta(\bhx,\bP) - \eta(\bx,\bP)$  at the $t$th iteration. 
The integration in~(\ref{eq:prob.integration})  is analytically intractable and expensive to compute numerically. Especially, when all $L$ input distributions are modeled nonparametrically,~(\ref{eq:prob.integration}) is an integration over a $\sum_{\ell=1}^L u_\ell$-dimensional probability simplex. Instead, we can sample $\bP_1,\bP_2,\ldots,\bP_B \sim \pi(\bP|\bz)$ and approximate~(\ref{eq:prob.integration}) with its MC estimate. The resulting risk set estimator is
\begin{equation}
\label{eq:ShatMC}
\bar{S}^{t}_\alpha(\delta) \equiv \left\{\bx\in\cX\left|\frac{1}{B}\sum_{b=1}^B
\Phi\left(\frac{\mu_t(\bhx,\bP_b)-\mu_t(\bx,\bP_b)-\delta}{\sigma_t(\bhx,\bx,\bP_b)}\right)> \alpha\right. \right\}.
\end{equation}
For notational convenience, we use $\Sbar{t}$ instead of $\bar{S}^{t}_\alpha(\delta)$ in the remainder of the paper  assuming $\alpha$ and $\delta$ are fixed throughout the procedure. 
While $\bP_1,\bP_2,\ldots,\bP_B$ may be newly sampled at every iteration, we choose to sample them  once at the very beginning of our sequential risk set estimation algorithm and use the same sample throughout the algorithm. 
The former is computationally expensive as we need to recompute the covariance vector $\Sigma_t(\bx,\bP_b)$ for all new $(\bx,\bP_b)$ combinations at each iteration to update the GP posterior as in~(\ref{eq:GPposterior}). By fixing $\bP_1,\bP_2,\ldots,\bP_B$, we only need to compute at most one additional element of $\Sigma_{t+1}(\bx,\bP_b)$ from $\Sigma_t(\bx,\bP_b)$ that corresponds to the newly simulated solution-distributions pair. When nonparametric input models are used, such computational saving indeed makes a difference since nonparametric kernel~(\ref{eq:Pkernel}) is more expensive to compute than a parametric one. 

For our sequential procedure, we consider the following class of sampling decision at each iteration for some function $h(\cdot,\cdot)$ defined by a sampling criterion:
\begin{equation}
\label{eq:sampling.decision}
{\arg\min}_{(\bx,\bP)\in \cX\times \{\bP_1, \bP_2, \ldots, \bP_B\}} h(\bx,\bP).
\end{equation}
A good sampling criterion reduces the estimation error of $\Sbar{t}$ efficiently and is cheap to compute.
A sampling criterion that ensures good global fit of the GP model (e.g.\ minimizing the integrated mean squared error of GP) does not directly target the former and can be quite inefficient for the same reason as for the naive MC approach; some solutions may not require too much simulation effort to be excluded from the risk set. 
In the following sections, we introduce a sampling criterion that directly targets reducing the estimation error of $\Sbar{t}$ and discuss its efficient computation. 

\subsection{Sequential sampling criterion}
\label{subsec:criterion}
To measure the estimation error of $\Sbar t$, we define loss function $\sL(\cdot)$ that counts the number of incorrectly classified solutions given a risk set estimate:
\begin{equation}
\label{eq:CICC}
\sL(\Sbar t) = \sum\nolimits_{\bx\in\S}I\{\bx \notin \Sbar t\}+\sum\nolimits_{\bx\notin\S}I\{\bx \in \Sbar t\}.
\end{equation}
We use $\mathcal{I}_t$ to denote the solution-distributions pair that is selected for simulation at the $t$th iteration. 
Suppose $\mathcal{I}_t = (\bx,\bP)$. The difference in the loss function between the $t$th and $(t+1)$th iteration is, $\Delta \sL^t(\bx,\bP) \equiv \sL(\Sbar{t+1})-\sL(\Sbar t)$, is 
\begin{align*}
&\Delta \sL^t(\bx,\bP)=\sum\nolimits_{\bx\in\S}\left(I\{\bx \notin \Sbar{t+1}\}-I\{\bx \notin \Sbar t\}\right)
+\sum\nolimits_{\bx\notin\S}\left(I\{\bx \in \Sbar{t+1}\} - I\{\bx \in \Sbar t\}\right)\\
&= |\S \cap (\Sbar t \backslash \Sbar{t+1})| - |\S \cap (\Sbar{t+1} \backslash \Sbar{t})|
+ |\S^c \cap (\Sbar{t+1} \backslash \Sbar{t})| - |\S^c \cap (\Sbar t \backslash \Sbar{t+1})|.
\end{align*}
Our strategy is to simulate $(\bx,\bP)$ expected to reduce the loss function the most in the next iteration. This type of sampling criterion is referred to as a one-step look-ahead or myopic policy as its goal is to maximize the benefit of sampling in the next iteration instead of the cumulative benefit in the subsequent iterations until the simulation budget is exhausted by formulating a dynamic programming problem~\citep{frazierthesis}. The latter is computationally impossible in our context due to so-called ``curse of dimensionality.''

To compute the expectation of $\Delta \sL^t(\bx,\bP),$ we need $\S$, which is clearly unknown. Instead, a lower bound to  $\Delta \sL^t(\bx,\bP)$ can be derived without knowing $\S$ as: 
\small
\begin{align}
\Delta \sL^t(\bx,\bP) 
& \geq -|\S \cap (\Sbar t \backslash \Sbar{t+1})| - |\S \cap (\Sbar{t+1} \backslash \Sbar{t})|
- |\S^c \cap (\Sbar{t+1} \backslash \Sbar{t})| - |\S^c \cap (\Sbar t \backslash \Sbar{t+1})| \nonumber  \\
& = -|(\Sbar t\backslash \Sbar{t+1})\cup(\Sbar{t+1}\backslash \Sbar{t})|.
\label{eq:delta.CC}
\end{align}
\normalsize
Note that $(\Sbar t\backslash \Sbar{t+1})\cup(\Sbar{t+1}\backslash \Sbar{t})$ is the set of solutions whose classifications are changed in the $(t+1)$th iteration. The lower bound can be also viewed as a `plug-in' estimate of $\Delta \sL^t(\bx,\bP)$ by letting $\S \approx \Sbar{t+1}$, which gives $\Delta \sL^t(\bx,\bP) \approx$(\ref{eq:delta.CC}).
From~(\ref{eq:delta.CC})
\begin{equation}
\label{eq:CICC.LB}
\E[\Delta \sL^t(\bx,\bP)|\bz,\sY_t,\cI_t= (\bx,\bP)]\geq  -\E\left[\left.|(\Sbar t\backslash \Sbar{t+1})\cup(\Sbar{t+1}\backslash \Sbar{t})|\right|\bz,\sY_t,\cI_t= (\bx,\bP)\right].
\end{equation}
Notice that the expectation is conditional on $\bz$ and $\sY_t$ as well as the sampling decision at the $t$th iteration, $\mathcal{I}_t$. 
Minimizing the lower bound in~(\ref{eq:CICC.LB}) is equivalent to  maximizing the expected number of solutions whose classifications change from the previous iteration. In other words, we would like to sample $(\bx,\bP)$ such that the updated risk set is expected to be \emph{as different as possible from the current to hedge the risk of incorrect classification.}

To estimate the lower bound in~(\ref{eq:CICC.LB}), we start with the following equivalence
\begin{align}
&\E\left[\left.|(\Sbar t\backslash \Sbar{t+1})\cup(\Sbar{t+1}\backslash \Sbar{t})|\right|\bz,\sY_t,\cI_t= (\bx,\bP)\right] \nonumber \\ 
&=\E\left[\left.
\sum_{\bx^\prime\in\Sbar t} I(\bx^\prime \notin \Sbar{t+1})
+ \sum_{\bx^\prime\notin\Sbar t} I(\bx^\prime \in \Sbar{t+1})
\right|\bz,\sY_t,\cI_t= (\bx,\bP)\right] \nonumber \\
&= \sum_{\bx^\prime\in\Sbar t} \Pr\{\bx^\prime \notin \Sbar{t+1}| \bz,\sY_t, \cI_t= (\bx,\bP)\} 
+ \sum_{\bx^\prime\notin\Sbar t} \Pr\{\bx^\prime \in \Sbar{t+1}| \bz,\sY_t, \cI_t= (\bx,\bP)\}\nonumber \\
&=\sum_{\bx^\prime\in\Sbar t} \Pr\bigg\{\frac{1}{B}\sum_{b=1}^B
\Phi\left(\frac{\mu_{t+1}(\bhx,\bP_b)-\mu_{t+1}(\bx^\prime,\bP_b)-\delta}{\sigma_{t+1}(\bhx,\bx^\prime,\bP_b)}\right)\leq \alpha\bigg| \sY_t, \cI_t= (\bx,\bP) \bigg\} \nonumber \\
&\;\;\;\; + \sum_{\bx^\prime\notin\Sbar t} \Pr\bigg\{
\frac{1}{B}\sum_{b=1}^B
\Phi\left(\frac{\mu_{t+1}(\bhx,\bP_b)-\mu_{t+1}(\bx^\prime,\bP_b)-\delta}{\sigma_{t+1}(\bhx,\bx^\prime,\bP_b)}\right)> \alpha
 \bigg|\sY_t, \cI_t= (\bx,\bP)\bigg\}, \label{eq:sum.prob}
\end{align}
where the last equality exploits the definition of the risk set estimator in~(\ref{eq:S.est}). Notice that~(\ref{eq:sum.prob}) can be interpreted as the sum of probabilities such that each solution $\bx^\prime$ switches sides from the risk set to its complement or \textit{vice versa}. 
Let $\boldsymbol{\mu}_t\in\mathbb{R}^{|\cX|B}$ and $\mathbf{V}_t\in \mathbb{R}^{|\cX|B\times |\cX|B}$ denote the mean vector and the variance-covariance matrix of the posterior GP at all $\cX\times \{\bP_1,\bP_2,\ldots,\bP_B\}$ at the $t$th iteration. The probabilities in~(\ref{eq:sum.prob}) are completely determined by the distribution of $\boldsymbol{\mu}_{t+1}$ and $\mathbf{V}_{t+1}$ given the sampling decision, $\cI_t= (\bx,\bP)$. 
In Appendix D, we show  
\begin{equation}
\label{eq:pred.mean}
\boldsymbol{\mu}_{t+1}|\bz,\sY_t,\cI_t= (\bx,\bP) \sim N\left(\boldsymbol{\mu}_t, \frac{\bV_t(\bx,\bP)\bV_t(\bx,\bP)^
\top}{v(\bx,\bP)/R_t + V_t(\bx,\bP;\bx,\bP)}\right),
\end{equation}
where $R_t$ is the number of replications we obtain at the selected solution-distributions pair at the $t$th iteration, and $\bV_t(\bx,\bP)$ and $V_t(\bx,\bP; \bx^\prime,\bP^\prime)$ are the column and the element of $\bV_t$ corresponding to $(\bx,\bP)$ and $\{(\bx,\bP),(\bx^\prime,\bP^\prime)\}$, respectively. Recall that $v(\bx,\bP)$ is the stochastic error variance of $Y(\bx,\bP)$ given $\bP$. If $(\bx,\bP)$ is already simulated, then $S^2(\bx,\bP)$ can be used as a plug-in estimate. Otherwise, one can fit a prediction model for $v(\bx,\bP)$ based on observed sample variances up to the $t$th iteration. For the experiments in Section~\ref{sec:expr},  we use simple pooled variances; for each $\bx$, $v(\bx,\bP)$ is approximated by the average of $S^2(\bx,\bP_b)$ at $(\bx,\bP_b)$ pairs that are simulated so far.  
The predictive variance-covariance matrix, $\bV_{t+1}$, does not depend on the simulation output from $(\bx,\bP)$ and only depends on the identity of $(\bx,\bP)$. We show in Appendix D, $\sigma_{t+1}(\bhx,\bx^\prime,\bP_b)$ given $\cI_t= (\bx,\bP)$ can be computed deterministically as 
\begin{align}
\label{eq:pred.var}
\sigma_{t+1}^2(\bhx,\bx^\prime,\bP_b)&|\bz,\sY_t,\cI_t = (\bx,\bP)  \;\;
= \;\;
\sigma^2_{t}(\bhx,\bx^\prime,\bP_b) - \frac{\left(V_t(\bx,\bP;\bhx,\bP_b) - V_t(\bx,\bP;\bx^\prime,\bP_b)\right)^2}{v(\bx,\bP)/R_t + V_t(\bx,\bP;\bx,\bP)}.
\end{align}

Analytically computing the probabilities in~(\ref{eq:sum.prob}) is difficult. A similar issue often arises in BO, where the sampling criterion is difficult to compute analytically. A common approach is MC estimation; we can sample $N$ multivariate normal vectors from~(\ref{eq:pred.mean}), plug them instead of $\boldsymbol{\mu}_{t+1}$ in~(\ref{eq:sum.prob}) and compute their average. 
The challenge in our case, however, is that this estimate quickly becomes $0$ as $t$ increases for fixed $N$. Recall that~(\ref{eq:sum.prob}) is the sum of the probability each solution switching its classification in the next iteration. This becomes a rare event as $t$ increases, because our sampling criterion causes $\E[|(\Sbar t\backslash \Sbar{t+1})\cup(\Sbar{t+1}\backslash \Sbar{t})||\bz,\sY_t,\cI_t= (\bx,\bP)]\xrightarrow{a.s.} 0$ for all $(\bx,\bP)$ pairs as shown in Section~\ref{sec:procedure}. As a result, the estimated sampling criteria for all $(\bx,\bP)$ become $0$  unless we increase $N$ as $t$ increases, which implies growing computational cost.

Instead, we propose an approximation of~(\ref{eq:sum.prob}) that can be computed exactly (up to a numerical precision) is cheaper to compute than the MC estimation approach described above. Applying the first-order Taylor series approximation around $\boldsymbol{\mu}_{t+1} = \boldsymbol{\mu}_{t}$, 
\small
\begin{align}
&\frac{1}{B}\sum_{b=1}^B
\Phi\left(\frac{\mu_{t+1}(\bhx,\bP_b)-\mu_{t+1}(\bx^\prime,\bP_b)-\delta}{\sigma_{t+1}(\bhx,\bx^\prime,\bP_b)}\right) \approx 
\frac{1}{B}\sum_{b=1}^B 
\Phi\left(\frac{\mu_{t}(\bhx,\bP_b)-\mu_{t}(\bx^\prime,\bP_b)-\delta}{\sigma_{t+1}(\bhx,\bx^\prime,\bP_b)}\right) \nonumber\\
&\;\;\;+ \frac{1}{B}\sum_{b=1}^B \phi\left(\frac{\mu_{t}(\bhx,\bP_b)-\mu_{t}(\bx^\prime,\bP_b)-\delta}{\sigma_{t+1}(\bhx,\bx^\prime,\bP_b)}\right)
\frac{\mu_{t+1}(\bhx,\bP_b)-\mu_{t+1}(\bx^\prime,\bP_b)-(\mu_{t}(\bhx,\bP_b)-\mu_{t}(\bx^\prime,\bP_b))}{\sigma_{t+1}(\bhx,\bx^\prime,\bP_b)},
\label{eq:Taylor.approx}
\end{align} \normalsize
where $\phi(\cdot)$ is the standard normal probability density function (pdf).
Note that the right-hand-side of~(\ref{eq:Taylor.approx}) is simply a linear function of $\boldsymbol{\mu}_{t+1}$. 
Conditional on $\bz,\sY_t,$ and ${\cI_t= (\bx,\bP)}$, its distribution 
can be derived from~(\ref{eq:pred.mean}) as 
\begin{equation}
\label{eq:normal.dist}
N\left(\frac{1}{B}\sum_{b=1}^B 
\Phi\left(\frac{\mu_{t}(\bhx,\bP_b)-\mu_{t}(\bx^\prime,\bP_b)-\delta}{\sigma_{t+1}(\bhx,\bx^\prime,\bP_b)}\right) ,
\left(\mathbf{c}_{\bx^\prime}^\top \mathbf{w}_{\bx^\prime}(\bx,\bP)\right)^2\right),
\end{equation}
where $\mathbf{c}_{\bx^\prime}$ and $\mathbf{w}_{\bx^\prime}(\bx,\bP)$ are $B$-dimensional vectors whose $b$th elements are $\frac{1}{B\sigma_{t+1}(\bhx,\bx^\prime,\bP_b)}\phi\left(\frac{\mu_{t}(\bhx,\bP_b)-\mu_{t}(\bx^\prime,\bP_b)-\delta}{\sigma_{t+1}(\bhx,\bx^\prime,\bP_b)}\right)$ and $\frac{V_t(\bhx,\bP_b;\bx,\bP)-V_t(\bx^\prime,\bP_b;\bx,\bP)}{\sqrt{v(\bx,\bP)/R_t + V_t(\bx,\bP;\bx,\bP)}}$, respectively.
Therefore,~(\ref{eq:sum.prob}) is approximated by 
\small
\begin{align}
&\widehat{\E}\left[\left.|(\Sbar t\backslash \Sbar{t+1})\cup(\Sbar{t+1}\backslash \Sbar{t})|\right|\bz,\sY_t,\cI_t= (\bx,\bP)\right]  \label{eq:norm.cdf.approx} \\
&=\sum_{\bx^\prime \in \Sbar t} \Phi\left(
\frac{\alpha-\frac{1}{B}\sum_{b=1}^B 
\Phi\left(\frac{\mu_{t}(\bhx,\bP_b)-\mu_{t}(\bx^\prime,\bP_b)-\delta}{\sigma_{t+1}(\bhx,\bx^\prime,\bP_b)}\right)}{|\mathbf{c}_{\bx^\prime}^\top \mathbf{w}_{\bx^\prime}(\bx,\bP)|}
\right)
+ \sum_{\bx^\prime\notin\Sbar t} 
\Phi\left(
\frac{\frac{1}{B}\sum_{b=1}^B 
\Phi\left(\frac{\mu_{t}(\bhx,\bP_b)-\mu_{t}(\bx^\prime,\bP_b)-\delta}{\sigma_{t+1}(\bhx,\bx^\prime,\bP_b)}\right)-\alpha}{|\mathbf{c}^\top_{\bx^\prime} \mathbf{w}_{\bx^\prime}(\bx,\bP)|}
\right), \nonumber 
\end{align} \normalsize
which can be computed exactly up to a numerical precision. 

Since the classification decision for $\bx$ is based on the comparisons between $\mu_{t+1}(\bhx,\bP_b)$ and $\mu_{t+1}(\bx,\bP_b)$ at $b=1,2,\ldots,B$, we also consider  pairwise sampling of $(\bhx,\bP)$ and $(\bx,\bP)$, i.e., $\cI_t=\{(\bhx,\bP),(\bx,\bP), \bx \neq \bhx\}$. The lower bound on $\E[\Delta \sL^t(\bx,\bP)|\bz,\sY_t,\cI_t=\{(\bhx,\bP),(\bx,\bP), \bx \neq \bhx\}]$ and its equivalent expression can be derived similarly as in~(\ref{eq:CICC.LB}) and~(\ref{eq:sum.prob}), respectively with only one difference: the conditioning event changes from $\cI_t= (\bx,\bP)$ to $\cI_t=\{(\bhx,\bP),(\bx,\bP), \bx \neq \bhx\}$. 
Appendix D shows
\begin{align}
\label{eq:pred.mean2}
&\boldsymbol{\mu}_{t+1}|\bz,\sY_t,\cI_t=\{(\bhx,\bP),(\bx,\bP), \bx \neq \bhx\}  \sim N\left(
\boldsymbol{\mu}_t, 
\mathbf{C}
\mathbf{D}^{-\top}\mathbf{D}^{-1}
\mathbf{C}^\top
\right), \\
&\sigma^2_{t+1}(\bhx,\bx^\prime,\bP_b) |\bz,\sY_t,\cI_t=\{(\bhx,\bP),(\bx,\bP), \bx \neq \bhx\} \nonumber\\
&=\sigma^2_{t}(\bhx,\bx^\prime,\bP_b) - \mathbf{d}(\bx^\prime,\bP_b;\bx,\bP)^\top
\mathbf{D}^{-\top}\mathbf{D}^{-1}
\mathbf{d}(\bx^\prime,\bP_b;\bx,\bP), \label{eq:pred.var2}
\end{align}
where $\mathbf{C} = \left[
\sqrt{\frac{R_t}{v(\bhx,\bP)}}\bV_t(\bhx,\bP), \sqrt{\frac{R_t}{v(\bx,\bP)}}\bV_t(\bx,\bP)
\right]$, 
$\mathbf{D}$ is the lower Cholesky factor of 
\[
\left[
\begin{matrix}
1+ R_tV_t(\bhx,\bP;\bhx,\bP)/v(\bhx,\bP) &  R_tV_t(\bhx,\bP;\bx,\bP)/\sqrt{v(\bx,\bP)v(\bhx,\bP)} \\
R_tV_t(\bhx,\bP;\bx,\bP)/\sqrt{v(\bx,\bP)v(\bhx,\bP)} & 1+ R_tV_t(\bx,\bP;\bx,\bP)/v(\bx,\bP)
\end{matrix}
\right],\mbox{ and}
\]
\[
\mathbf{d}(\bx^\prime,\bP_b;\bx,\bP) = \left[
\begin{matrix}
\sqrt{\frac{R_t}{v(\bhx,\bP)}}(V_t(\bhx,\bP_b;\bhx,\bP)-V_t(\bx^\prime,\bP_b;\bhx,\bP))\\ 
\sqrt{\frac{R_t}{v(\bx,\bP)}}(V_t(\bhx,\bP_b;\bx,\bP)-V_t(\bx^\prime,\bP_b;\bx,\bP))
\end{matrix}
\right].
\] 
From~(\ref{eq:pred.mean2}), the predictive distribution of~(\ref{eq:Taylor.approx}) conditional on $\bz,\sY_t,$ and ${\cI_t= \{(\bhx,\bP),(\bx,\bP)\}}$ can be derived; it has the same mean as~(\ref{eq:normal.dist}), but the variance is 
$\mathbf{c}_{\bx^\prime}^\top \mathbf{U}_{\bx^\prime}(\bx,\bP)\mathbf{U}_{\bx^\prime}(\bx,\bP)^\top \mathbf{c}_{\bx^\prime},$ where $\mathbf{U}_{\bx^\prime}(\bx,\bP)$ is a $B\times 2$ matrix whose $b$th row is $\mathbf{d}(\bx^\prime,\bP_b;\bx,\bP)^\top\mathbf{D}^{-\top}$. Thus, 
$\widehat{\E}\left[\left.|(\Sbar t\backslash \Sbar{t+1})\cup(\Sbar{t+1}\backslash \Sbar{t})|\right|\bz,\sY_t,\cI_t= \{(\bhx,\bP),(\bx,\bP)\}\right]$ can be defined by replacing $\mathbf{c}_{\bx^\prime}^\top \mathbf{w}_{\bx^\prime}(\bx,\bP)$ in~(\ref{eq:norm.cdf.approx}) with $(\mathbf{c}_{\bx^\prime}^\top \mathbf{U}_{\bx^\prime}(\bx,\bP)\mathbf{U}_{\bx^\prime}(\bx,\bP)^\top \mathbf{c}_{\bx^\prime})^{1/2}$. 

Based on the estimated lower bounds of expected reduction in the loss function, we decide which $(\bx,\bP)$ (or $(\bhx,\bP)$ and $(\bx,\bP)$, if pair-wise sampling) to simulate next. Specifically, we solve~(\ref{eq:sampling.decision}) with 
\small
\begin{align*}
\begin{array}{l}
h(\bx,\bP) = \left\{
\begin{array}{l l}
-\widehat{\E}\left[\left.|(\Sbar t\backslash \Sbar{t+1})\cup(\Sbar{t+1}\backslash \Sbar{t})|\right|\bz,\sY_t,\cI_t= (\bx,\bP)\right], & \mbox{ if } \bx = \bhx, \\
\begin{array}{l}
 -\max\bigg\{\widehat{\E}\left[\left.|(\Sbar t\backslash \Sbar{t+1})\cup(\Sbar{t+1}\backslash \Sbar{t})|\right|\bz,\sY_t,\cI_t= (\bx,\bP)\right], \\
 \;\;\;\;\;\;\;\;\;\;\;\; \frac{1}{2}\widehat{\E}\left[\left.|(\Sbar t\backslash \Sbar{t+1})\cup(\Sbar{t+1}\backslash \Sbar{t})|\right|\bz,\sY_t,\cI_t=\{(\bhx,\bP),(\bx,\bP)\}\right]\bigg\}, 
 \end{array}
 & \mbox{ otherwise}.
\end{array}
 \right.
\end{array}
\end{align*}
\normalsize
Notice for pairwise sampling, we discount the reduction by a half since it requires twice the simulation effort.

\subsection{Distribution selection criterion}
\label{subsec:dist.selection}
Although we select the next $\bP$ to simulate among $\bP_1,\bP_2,\ldots,\bP_B$, we need to evaluate $h$ for $|\cX|B$ solution-distributions pairs to compare at each iteration. 
To reduce the computational burden, first notice that Problem~(\ref{eq:sampling.decision}) can be rewritten as $\min_{\bx \in \cX}\min_{\bP \in \{\bP_1,\bP_2,\ldots,\bP_B\}} h(\bx,\bP)$. Thus, if we can analytically solve the inner minimization problem for each $\bx$, then $h$ needs to be evaluated only $|\cX|$ times. Even though this is not possible, we can approximate the minimizer of the inner problem by a ``good'' distribution for each $\bx$. 

What characterizes a good distribution to be sampled with $\bx$?  For $\bx=\bhx,$ we argue that 
\begin{equation}
\label{eq:dist.xhat}
\bP^1_{\bhx}\equiv {\arg\max}_{\bP_b\in\{\bP_1,\bP_2,\ldots,\bP_B\}}V_t(\bhx,\bP_b;\bhx,\bP_b)
\end{equation}
is a good choice because the GP has the largest prediction error at $(\bhx,\bP^1_{\bhx})$ given $\bhx$. Note that ``$1$'' in $\bP^1_{\bhx}$ is to emphasize that only single sampling is considered for $\bhx$. 

For $\bx\neq\bhx$, in the same spirit of our sampling criterion in Section~\ref{subsec:criterion}, we would like to sample $\bP_b$ that is most likely to change the classification decision of $\bx$ in the next iteration to hedge the risk of misclassification.
Recall that each $\bP_b$ contributes $\frac{1}{B}\Phi\left(\frac{\mu_t(\bhx,\bP_b)-\mu_t(\bx,\bP_b)-\delta}{\sigma_t(\bhx,\bx,\bP_b)}\right)$ to the MC estimate of (\ref{eq:prob.integration}). Thus, the magnitude of  `local'  change in contribution when $\cI_t = (\bx,\bP_b)$ (or $\cI_t=\{(\bhx,\bP_b),(\bx,\bP_b)\}$) is 
\begin{equation}
\label{eq:local.contribution}
\frac{1}{B}\left|\Phi\left(\frac{\mu_{t+1}(\bhx,\bP_b)-\mu_{t+1}(\bx,\bP_b)-\delta}{\sigma_{t+1}(\bhx,\bx,\bP_b)}\right) - \Phi\left(\frac{\mu_t(\bhx,\bP_b)-\mu_t(\bx,\bP_b)-\delta}{\sigma_t(\bhx,\bx,\bP_b)}\right)\right|.
\end{equation}
We argue $\bP_b$ that maximizes the expected value of~(\ref{eq:local.contribution}) is a good candidate because the larger the expected change is, the more likely the classification of $\bx$ changes. 
Clearly, this is a local change as sampling $(\bx,\bP_b)$
will affect the GP means and variances at all solution-distributions pairs. However, focusing on the local change allows us to estimate its 
expectation cheaply since it only requires $\sigma_{t+1}(\bhx,\bx,\bP_b)$ and the predictive distribution of $\mu_{t+1}(\bhx,\bP_b)-\mu_{t+1}(\bx,\bP_b)$ instead of the entire $\boldsymbol{\mu}_{t+1}$.
The exact expectation of~(\ref{eq:local.contribution}) is difficult to obtain analytically. Instead, we can approximate it in a similar way to~(\ref{eq:Taylor.approx}). The first-order Taylor series approximation of~(\ref{eq:local.contribution}) at $\mu_{t+1}(\bhx,\bP_b) - \mu_{t+1}(\bx,\bP_b) = \mu_{t}(\bhx,\bP_b) - \mu_{t}(\bx,\bP_b)$ gives
\small
\begin{align}\label{eq:Taylor2}
&\frac{1}{B}\Bigg|\Phi\left(\frac{\mu_{t}(\bhx,\bP_b)-\mu_{t}(\bx,\bP_b)-\delta}{\sigma_{t+1}(\bhx,\bx,\bP_b)}\right) - \Phi\left(\frac{\mu_t(\bhx,\bP_b)-\mu_t(\bx,\bP_b)-\delta}{\sigma_t(\bhx,\bx,\bP_b)}\right) \\
&+\frac{1}{\sigma_{t+1}(\bhx,\bx,\bP_b)}\phi\left(\frac{\mu_{t}(\bhx,\bP_b)-\mu_{t}(\bx,\bP_b)-\delta}{\sigma_{t+1}(\bhx,\bx,\bP_b)}\right)
\left\{\mu_{t+1}(\bhx,\bP_b)-\mu_{t+1}(\bx,\bP_b) - \left(\mu_{t}(\bhx,\bP_b)-\mu_{t}(\bx,\bP_b)\right)
\right\}
\Bigg|, \nonumber
\end{align}
\normalsize
which is a folded normal random variable. 
Therefore, the expectation of~(\ref{eq:Taylor2}) can be derived from the following mean formula for a folded normal distribution~\citep{foldednormal}:
\begin{equation}
\label{eq:folded.normal}
\tfrac{1}{B}\left\{(1-2\Phi(-{a_1}/{a_2}))a_1 + 2a_2\phi(-{a_1}/{a_2})\right\},
\end{equation}
where $a_1$ and $a_2$ are the mean and the standard deviation of the normal random variable inside $|\cdot|$ in~(\ref{eq:Taylor2}), respectively. For 
$\cI_t = (\bx,\bP_b)$, $a_1 = \Phi\left(\frac{\mu_{t}(\bhx,\bP_b)-\mu_{t}(\bx,\bP_b)-\delta}{\sigma_{t+1}(\bhx,\bx,\bP_b)}\right) - \Phi\left(\frac{\mu_t(\bhx,\bP_b)-\mu_t(\bx,\bP_b)-\delta}{\sigma_t(\bhx,\bx,\bP_b)}\right)$ and 
$a_2 = \phi\left(\frac{\mu_{t}(\bhx,\bP_b)-\mu_{t}(\bx,\bP_b)-\delta}{\sigma_{t+1}(\bhx,\bx,\bP_b)}\right)\frac{|V_t(\bx,\bP;\bhx,\bP_b) - V_t(\bx,\bP;\bx,\bP_b)|}{\sigma_{t+1}(\bhx,\bx,\bP_b)\sqrt{v(\bx,\bP_b)/R_t + V_t(\bx,\bP_b;\bx,\bP_b)}}$, where $\sigma_{t+1}(\bhx,\bx,\bP_b)$ is given by~(\ref{eq:pred.var}). For pairwise sampling, ${\cI_t=\{(\bhx,\bP_b),(\bx,\bP_b),\bx\neq\bhx\}}$, $a_1$ has the same expression, but with $\sigma_{t+1}(\bhx,\bx,\bP_b)$ in~(\ref{eq:pred.var2}), and $a_2 = \phi\left(\frac{\mu_{t}(\bhx,\bP_b)-\mu_{t}(\bx,\bP_b)-\delta}{\sigma_{t+1}(\bhx,\bx,\bP_b)}\right)\frac{\sqrt{\mathbf{d}(\bx,\bP_b;\bx,\bP_b)^\top
\mathbf{D}^{-\top}\mathbf{D}^{-1}
\mathbf{d}(\bx,\bP_b;\bx,\bP_b)}}{\sigma_{t+1}(\bhx,\bx,\bP_b)}$, where $\mathbf{D}$ and $\mathbf{d}(\bx,\bP_b;\bx,\bP_b)$ are defined in Section~\ref{subsec:criterion}.
We define $\mathcal{H}^1_\bx(\bP_b)$ and $\mathcal{H}^2_\bx(\bP_b)$ as the resulting function~(\ref{eq:folded.normal}) for single and pairwise sampling cases, respectively. 

Let $\bP_\bx^1 = \arg\max_{\bP_b \in \{\bP_1,\bP_2,\ldots,\bP_B\}}\mathcal{H}^1_\bx(\bP_b)$ and $\bP_\bx^2 = \arg\max_{\bP_b \in \{\bP_1,\bP_2,\ldots,\bP_B\}}\mathcal{H}^2_\bx(\bP_b)$ for $\bx\neq\bhx$. In words, $\bP_\bx^1$ $(\bP_\bx^2)$ is the best input distribution to simulate with $\bx$ for single (pair-wise) sampling. 
Given $\bP_\bx^1$ and $\bP_\bx^2$ for all $\bx\in\cX$,~(\ref{eq:sampling.decision}) is modified to $\min_{\bx \in \cX} \tilde{h}(\bx)$, where
\small
\begin{align*}
\begin{array}{l}
\tilde{h}(\bx) = \left\{
\begin{array}{l l}
-\widehat{\E}\left[\left.|(\Sbar t\backslash \Sbar{t+1})\cup(\Sbar{t+1}\backslash \Sbar{t})|\right|\bz,\sY_t,\cI_t= (\bx,\bP_\bx^1)\right], & \mbox{ if } \bx = \bhx, \\
\begin{array}{l}
 -\max\bigg\{\widehat{\E}\left[\left.|(\Sbar t\backslash \Sbar{t+1})\cup(\Sbar{t+1}\backslash \Sbar{t})|\right|\bz,\sY_t,\cI_t= (\bx,\bP_\bx^1)\right], \\
 \;\;\;\;\;\;\;\;\;\;\;\; \frac{1}{2}\widehat{\E}\left[\left.|(\Sbar t\backslash \Sbar{t+1})\cup(\Sbar{t+1}\backslash \Sbar{t})|\right|\bz,\sY_t,\cI_t=\{(\bhx,\bP_\bx^2),(\bx,\bP_\bx^2)\}\right]\bigg\}, 
 \end{array}
 & \mbox{ otherwise}.
\end{array}
 \right.
\end{array}
\end{align*}
\normalsize
Thus, $\tilde{h}$ needs to be evaluated only $|\cX|$ times once $\bP_\bx^1$ and $\bP_\bx^2$ are found for each $\bx$.

The distribution selection problem discussed in this section can be interpreted as choosing a point on the support ($\bP_b$) to sample to estimate the probability that a $\eta(\bhx,\bP)-\eta(\bx,\bP)$ is above threshold $\delta$.  This is closely related to the Bayesian superset estimation  reviewed in Section~\ref{sec:lit}. For instance, the sequential sampling criteria proposed by~\cite{Bect2012} can be applied to the distribution selection problem. However, their criteria are more expensive to compute than ours as they measure the effect of sampling each candidate $\bP_b$ to all $\bP_1,\bP_2,\ldots,\bP_B$ and involve a numerical integration. Since we need to solve this problem twice (for single and pairwise sampling) for each $\bx\in\cX$, such computational overhead is undesirable. 

In the same setting, \cite{echard2011} propose to sample $\bP_b\in\{\bP_1,\bP_2,\ldots,\bP_B\}$ such that the marginal event,
$\{\eta(\bhx,\bP_b)-\eta(\bx,\bP_b)>\delta\}$, is the most uncertain. This is equivalent to choosing $\bP_b$ with the smallest $|\delta-(\mu_t(\bhx,\bP_b)-\mu_t(\bx,\bP_b))|/\sigma_t(\bhx,\bx,\bP_b)$. Since this does not involve computing the predictive distribution, it is cheaper to compute than our criterion. However, this criterion can be quite inefficient if there exists $\bP_{b^\prime}$ such that $\E[Y(\bhx,\bP_{b^\prime})|\bP_{b^\prime}]-\E[Y(\bx,\bP_{b^\prime})|\bP_{b^\prime}]$ is close to $\delta$, which causes $\bP_{b^\prime}$ to get disproportionately large sampling effort. 
We compare the performance of this sampling criterion with ours in Section~\ref{subsec:mm1k}.

\section{Sequential risk set inference procedure}
\label{sec:procedure}

We present the sequential risk set inference (SRSI) procedure in Algorithm~\ref{alg:1}. The procedure selects the next solution-distributions pair to simulate using the sampling criteria introduced in Section~\ref{sec:inference}. 

\begin{algorithm} [tb]
\singlespacing
\caption{Sequential risk set inference procedure (SRSI)}
\label{alg:1}
\begin{algorithmic}[1]
\State Initialize $B, n_0, r$ and $\{R_t\}$ for the algorithm. 
\State From real-world data $\bz$, update the posterior distribution, $\pi(\bP|\bz)$.
\State Sample $\bP_1,\bP_2,\ldots, \bP_B$ from $\pi(\bP|\bz)$.  \label{step:boot}
\State Select $n_0$ initial $(\bx,\bP)$ pairs from $\bx\in\mathcal{\cX}$ and $\bP\in\{\bP_1,\bP_2,\ldots, \bP_B\}$, simulate $r$ replications at each to obtain $\mathsf{Y}_0$, and estimate the parameters of the GP prior via MLE. $t\leftarrow 0$.
\While {simulation budget remains} 
\State Update GP posterior mean, $\boldsymbol{\mu}_{t}$, and covariance matrix, $\mathbf{V}_t$, conditional on $\sY_t$. 
\State Compute $\sigma_t(\bhx,\bx,\bP)$ for all $\bx\neq\bhx$ and $\bP \in \{\bP_1,\bP_2,\ldots,\bP_B\}$. 
\State Find $\bar{S}^{t} \equiv \left\{\bx\in\cX\left|\frac{1}{B}\sum_{b=1}^B
\Phi\left(\frac{\mu_t(\bhx,\bP_b)-\mu_t(\bx,\bP_b)-\delta}{\sigma_t(\bhx,\bx,\bP_b)}\right)> \alpha\right. \right\}$.
\State Find $\bP^1(\bhx) =  {\arg\max}_{\bP_b \in \{\bP_1,\bP_2,\ldots,\bP_B\}} V_t(\bhx,\bP_b;\bhx,\bP_b)$. 
\State For each $\bx\neq\bhx$, find  $\bP^1(\bx) =  {\arg\max}_{\bP_b \in \{\bP_1,\bP_2,\ldots,\bP_B\}} \mathcal{H}^1_\bx(\bP_b)$ and
$\bP^2(\bx) =  {\arg\max}_{\bP_b \in \{\bP_1,\bP_2,\ldots,\bP_B\}} \mathcal{H}^2_\bx(\bP_b)$.
\State \label{step:opt} Find $\tilde{\bx} = \arg\min_{\bx\in\cX} \tilde{h}(\bx)$.
\If {$\tilde{\bx}=\bhx$}
\State $\tilde{\bP} = \bP^1(\tilde{\bx})$
\Else 
\If {$\left\{ \begin{array}{l}
\widehat{\E}\left[\left.|(\Sbar t\backslash \Sbar{t+1})\cup(\Sbar{t+1}\backslash \Sbar{t})|\right|\bz,\sY_t,\cI_t= (\tilde{\bx},\bP_{\tilde{\bx}}^1)\right]  \\
>\frac{1}{2}\widehat{\E}\left[\left.|(\Sbar t\backslash \Sbar{t+1})\cup(\Sbar{t+1}\backslash \Sbar{t})|\right|\bz,\sY_t,\cI_t=\{(\bhx,\bP_{\tilde{\bx}}^2),(\tilde{\bx},\bP_{\tilde{\bx}}^2)\}\right]
\end{array}\right\}$  }
\State $\tilde{\bP} = \bP^1_{\tilde{\bx}}$
\Else
\State $\tilde{\bP} = \bP^2_{\tilde{\bx}}$
\EndIf
\EndIf
\State Run $R_t$ replications at $(\tilde{\bx},\tilde{\bP})$ and update $\sY_{t+1}$. $t\leftarrow t+1$.
\EndWhile
\State Update GP posterior mean and covariance matrix conditional on $\sY_t$. Find and return $\Sbar t$.
\end{algorithmic}
\end{algorithm}

In the following, we discuss asymptotic properties of SRSI. 
Proofs of all theorems below can be found in Appendix E. Before we proceed, let us define the following estimator of $\S$: 
\[
\tilde{S}^B_\alpha(\delta)  \equiv \left\{\bx\in\cX\left|\frac{1}{B}\sum_{b=1}^B
I\left(\E[Y(\bhx;\bP_b)|\bP_b]-\E[Y(\bx;\bP_b)|\bP_b]>\delta\right)> \alpha\right. \right\},
\]
where $\bP_1,\bP_2,\ldots,\bP_B$ are the same as those from Step~(3) of the algorithm. In words, $\tilde{S}^B_\alpha(\delta)$ is a MC estimator of $\S$ assuming we can evaluate $\E[Y(\bx;\bP_b)|\bP_b]$ exactly without simulation error, i.e., the best-possible estimator if no analytical expression of $\E[Y(\bx;\bP_b)|\bP_b]$ is available. The following theorem shows that $\Sbar t$ indeed converges to $\tilde{S}^B_\alpha(\delta)$ as $t\to\infty$. 

\begin{theorem}
\label{thm:asymptotic} 
Suppose either all $L$ input distributions are nonparametrically modeled with the Dirichlet prior/posterior distribution or if there is  any parametric $P_\ell$,  its parameter $\theta_\ell$ has a continuous posterior distribution.
Also suppose $B$ is chosen so that $\alpha$ is not a multiple of $1/B$ and $\S$ is replaced with $\tilde{S}^B_\alpha(\delta)$ in the definition of $\sL$.
If SRSI runs  without stopping, then 
1) $\Sbar t \xrightarrow{a.s.} \tilde{S}^B_\alpha(\delta)$; and 2) $-|(\Sbar t\backslash \Sbar{t+1})\cup(\Sbar{t+1}\backslash \Sbar{t})|\big|\bz,\sY_t,\cI_t\xrightarrow{a.s.}\Delta \sL^t(\bx,\bP_b)$ for both $\cI_t=(\bx,\bP_b)$ and $\cI_t=\{(\bhx,\bP_b),(\bx,\bP_b),\bx\neq\bhx\}$ uniformly for all $(\bx,\bP_b)$ pairs. 
\end{theorem}
\noindent Note that the conditions on $\bP$ and $B$ in Theorem~\ref{thm:asymptotic} are needed to ensure there is no $\bx$ such that $\frac{1}{B}\sum_{b=1}^B
I\left(\E[Y(\bhx;\bP_b)|\bP_b]-\E[Y(\bx;\bP_b)|\bP_b]>\delta\right)=\alpha$. 
The second part of Theorem~\ref{thm:asymptotic} states that our lower bound to the exact one-step reduction in the loss function becomes asymptotically tight. In the proof of Theorem~\ref{thm:asymptotic}, we also show that  $\E[|(\Sbar t\backslash \Sbar{t+1})\cup(\Sbar{t+1}\backslash \Sbar{t})||\bz,\sY_t,\cI_t]\xrightarrow{a.s.} 0$ for both $\cI_t=(\bx,\bP_b)$ and $\cI_t=\{(\bhx,\bP_b),(\bx,\bP_b),\bx\neq\bhx\}$, which is the result mentioned in Section~\ref{subsec:criterion}.

Although SRSI uses $\bP_1,\bP_2,\ldots,\bP_B$ sampled in Step~\ref{step:boot} throughout the procedure for computational efficiency, at any finite $t$ the GP posterior can be evaluated at $\cX \times \{\bP_1,\bP_2,\ldots,\bP_{\tilde{B}}\}$ for $\tilde{B}>B$ to obtain a more accurate estimate of the risk set. The following theorem facilitates this.  
\begin{theorem}
\label{thm:Bincrease} Suppose SRSI is stopped  at $t=T$. For any $\bP_1,\bP_2,\ldots,\bP_{\tilde{B}}\overset{i.i.d.}{\sim} \pi(\bP|\bz)$, as $\tilde{B}\to\infty,$
$$\bigg\{\bx\in\cX\bigg|\tfrac{1}{\tilde{B}}\sum_{b=1}^{\tilde{B}}
\Phi\left(\tfrac{\mu_T(\bhx,\bP_b)-\mu_T(\bx,\bP_b)-\delta}{\sigma_T(\bhx,\bx,\bP_b)}\right)> \alpha \bigg\} \xrightarrow{a.s.} \left\{\bx\in\cX\left|\Pr\{\eta(\bhx,\bP)-\eta(\bx,\bP)>\delta|\sY_T, \bz\}> \alpha \right.\right\}.$$
\end{theorem}

Even though SRSI is designed to efficiently estimate the risk set for specific $\alpha$ and $\delta$, the posterior GP from the procedure can be used to estimate risk sets for other $\alpha$ and $\delta$ values, which we demonstrate empirically in Section~\ref{subsec:mm1k}.

\section{Empirical performance}
\label{sec:expr}
In this section, we show empirical performance of SRSI using two examples. The first is the M/M/$1$/$k$ queue example  introduced in Section~\ref{sec:riskset}. This problem has a known expression for the objective function, thus lets us evaluate the estimation error of the risk set obtained from SRSI. 
We compare the performance of SRSI with a naive MC estimation method using the same simulation budget. We also test two variations of the SRSI with different distribution selection criteria to show effectiveness of our criterion introduced in Section~\ref{subsec:dist.selection}.

The second is a more realistic example simplified from the ambulance dispatching center location problem introduced in Section~\ref{sec:intro}. Using this example, we illustrate practical usage of the risk set.

\subsection{$M/M/1/k$ queue example}
\label{subsec:mm1k}

As introduced in Section~\ref{sec:riskset}, the objective function of this problem is the expected cost per customer and $\bx=k$. Recall that $P_1^c$ and and $P_2^c$ are Exp($1$) and Exp$(1.1)$, respectively. 
The input distributions are modeled nonparametrically using the Dirichlet prior described in Section~\ref{subsec:inputmodeling} with the concentration parameter, $\kappa_1 = \kappa_2 = 1$. We implemented a discrete-event simulator that modifies the Lindley equation for a uncapacitated single server queue~\citep{nelsontextbook} to incorporate the system capacity. Although the simulator can generate interarrival times and service times directly from $\bP=\{P_1,P_2\}$  sampled from the posterior Dirichlet distribution, we chose to first compute the means of $P_1$ and $P_2$ and use them as parameters for exponential distributions to generate interarrival and service times. This facilitates performance evaluation presented below as we can use the exact expected cost expression in Appendix A. To remove the initial bias of the simulation, we sampled the initial number of customers in the system from the steady-state distribution of number in system for M/M/$1$/$k$ queue. The simulation output from a single replication given $\bx$ and $\bP$, $Y(\bx;\bP)$, is the average cost of the $2{,}000$ customers generated within the replication. 

The solutions in contention are $\cX = \{1,2,\ldots,50\}$ and we set $\delta = 1$ and $\alpha = 0.2$. For GP prior, we adopted the squared exponential kernel for $\gamma_{\cX}$ and nonparametric kernel~(\ref{eq:Pkernel}) for $\gamma_{\mathcal{M}}$ with the squared Hellinger distance as $D$. 

We compare the risk set estimates from SRSI with those from the following procedures:
\begin{itemize}
\item Naive MC (NMC): all $|\cX|B$ solution-distribution pairs are assigned the equal number of replications, $N$, given the total simulation budget and the risk set is estimated by
$
\left\{\bx\in\cX\left|\frac{1}{B}\sum_{b=1}^B
I(\bar{Y}(\bhx;\bP_b) - \bar{Y}(\bx;\bP_b) > \delta)> \alpha\right. \right\},
$
where $\bar{Y}(\bhx;\bP_b)$ is the sample average of $N$ replications at $(\bx,\bP_b)$.
\item SRSI-m: SRSI with modified $\mathcal{H}^1_\bx(\bP_b) = \mathcal{H}^2_\bx(\bP_b) = -\frac{|\delta-(\mu_t(\bhx,\bP_b)-\mu_t(\bx,\bP_b))|}{\sigma_t(\bhx,\bx,\bP_b)}$ for $\bx\neq\bhx$.
\item SRSI-v: SRSI with modified $\mathcal{H}^1_\bx(\bP_b) = \mathcal{H}^2_\bx(\bP_b) = -\sigma_t(\bhx,\bx,\bP_b)$ for $\bx\neq\bhx$.
\end{itemize}
The SRSI-m and SRSI-v are variations of SRSI, where the former selects $\tilde{\bP} = \bP_b$ minimizing the marginal sampling criterion suggested by~\cite{echard2011} and the latter selects $\bP_b$ with the largest posterior variance of $\eta_t(\bhx,\bP_b)-\eta_t(\bx,\bP_b)$, respectively, at each iteration. 

We repeated each procedure $120$ times while changing the random number seeds from $1$ to $120$. For each run, a new  ``real-world'' sample $\bz$ is generated by sampling $m=100$ i.i.d.\ random variates from each of $P_1^c$ and $P_2^c$. Since all procedures use the same random number seed for each run, they all share the same $\pi(\bP|\bz)$ for each run that varies across runs. We chose $\bhx$ to be the conditional optimum given the MAP of $\bP$ and sampled  $B=101$ probability simplices from $\pi(\bP|\bz)$. For SRSI and its variations, $n_0 = 100$ initial solution-distributions pairs are chosen and simulated $r=30$ times to compute the MLEs of the GP hyperparameters and we set $R_t = 30$ for all $t$.  For benchmarking, we found $\tilde{S}^B_\alpha(\delta)$ using the exact cost function in Appendix A. 

\begin{figure} [t]
\begin{subfigure} {.5\textwidth}
\begin{center}
\includegraphics[scale=0.57]{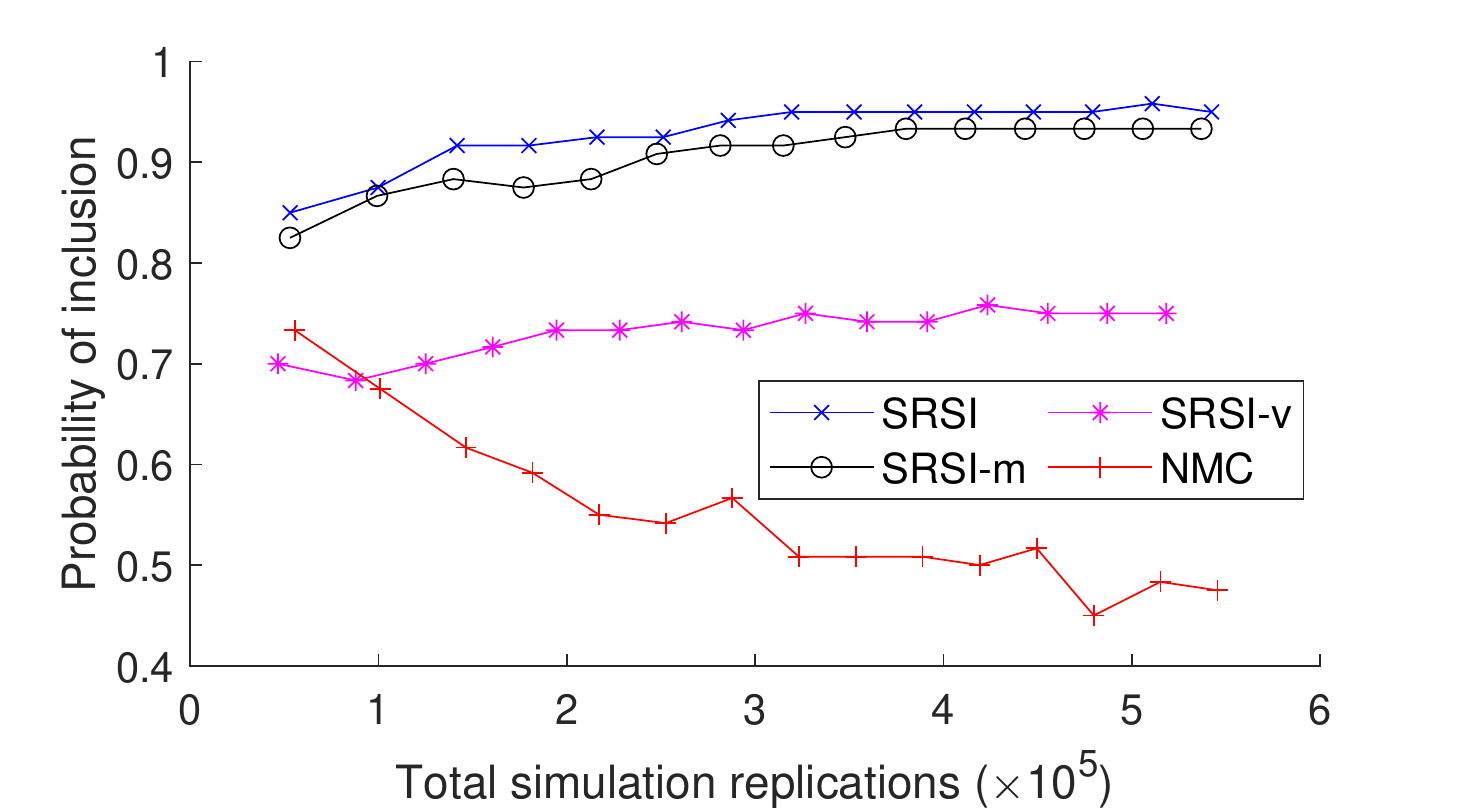}
\end{center}
\end{subfigure}%
\begin{subfigure} {.5\textwidth}
\begin{center}
\includegraphics[scale=0.57]{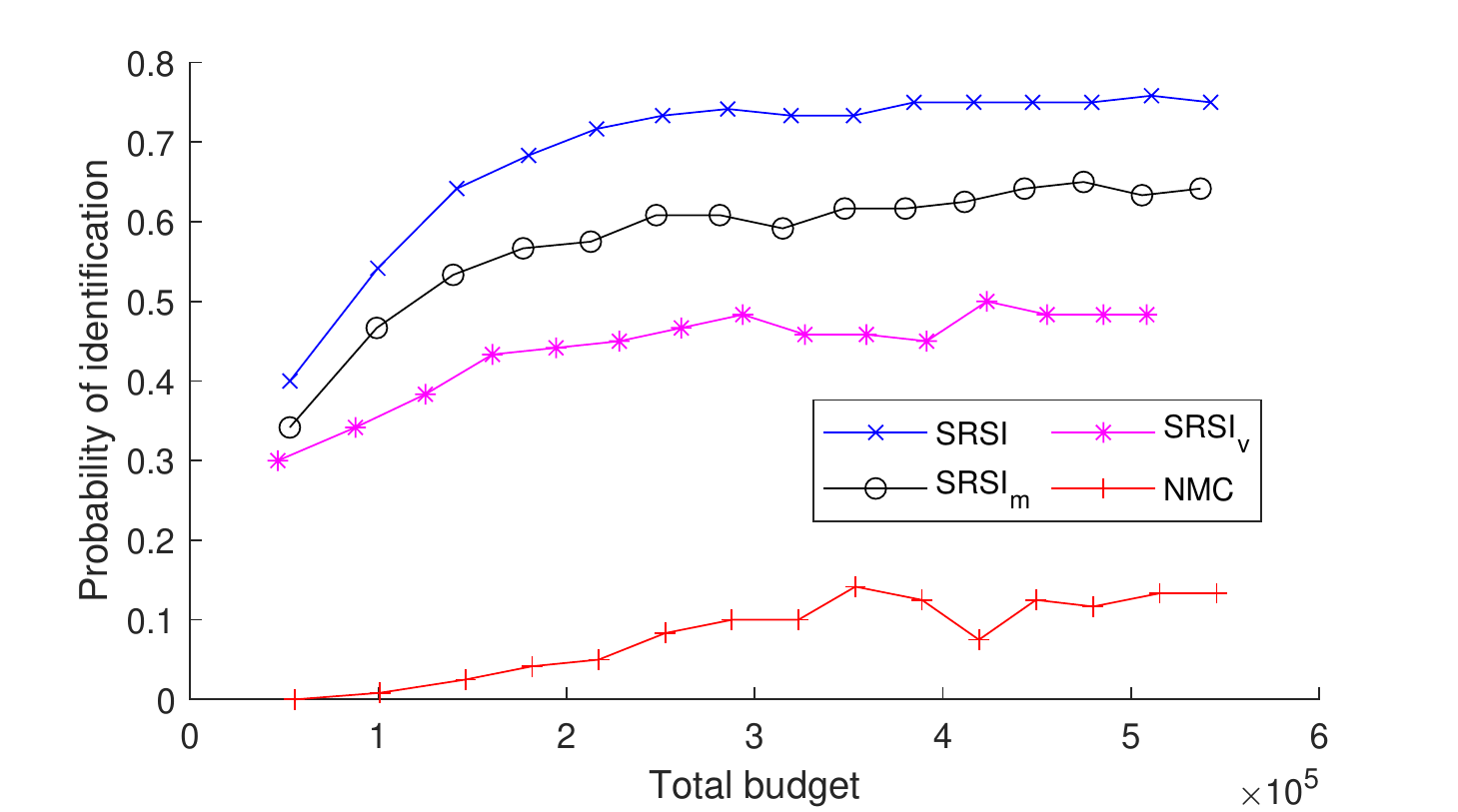}
\end{center}
\end{subfigure}
\begin{subfigure} {.5\textwidth}
\begin{center}
\includegraphics[scale=0.57]{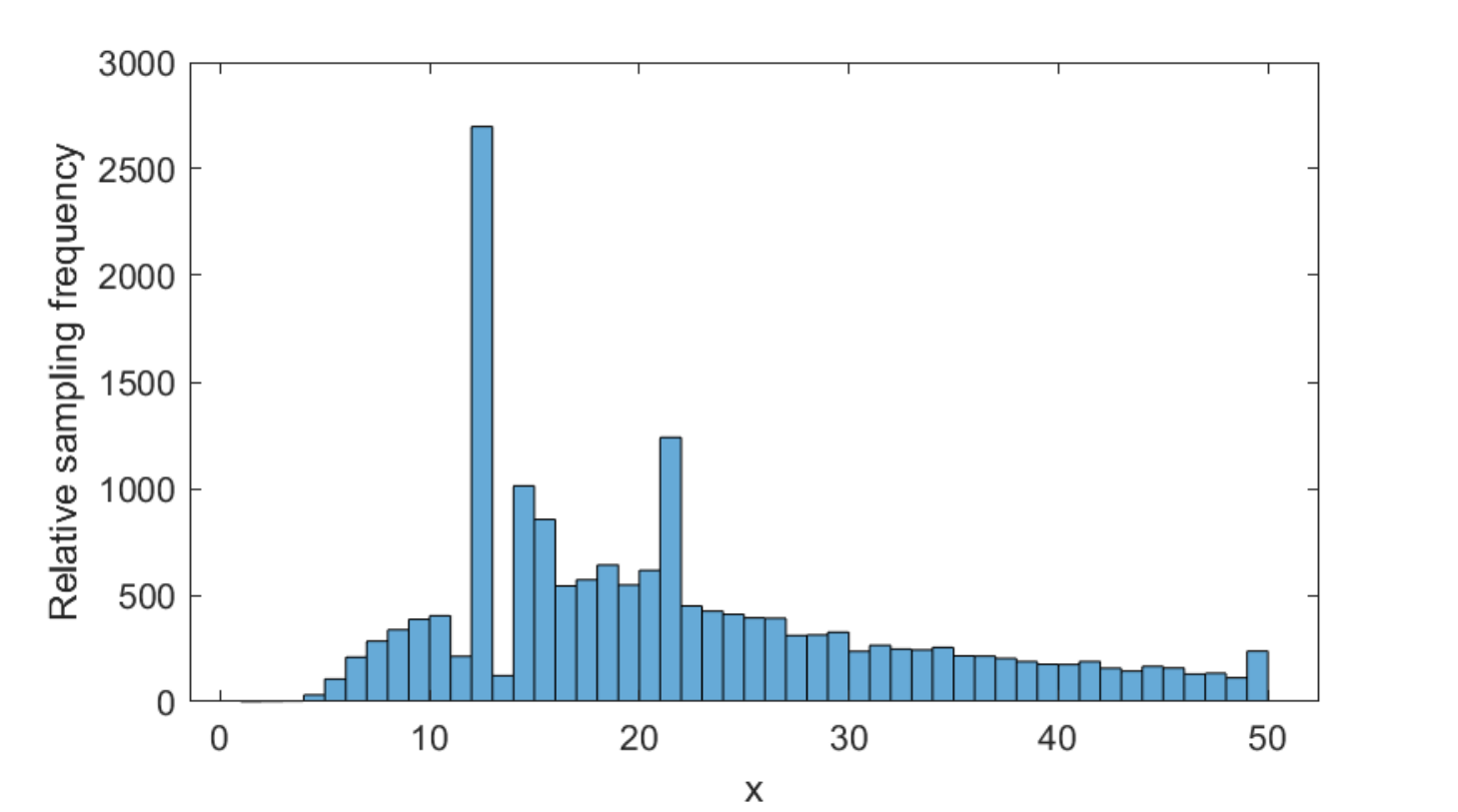}
\end{center}
\end{subfigure}
\begin{subfigure} {.5\textwidth}
\begin{center}
\includegraphics[scale=0.57]{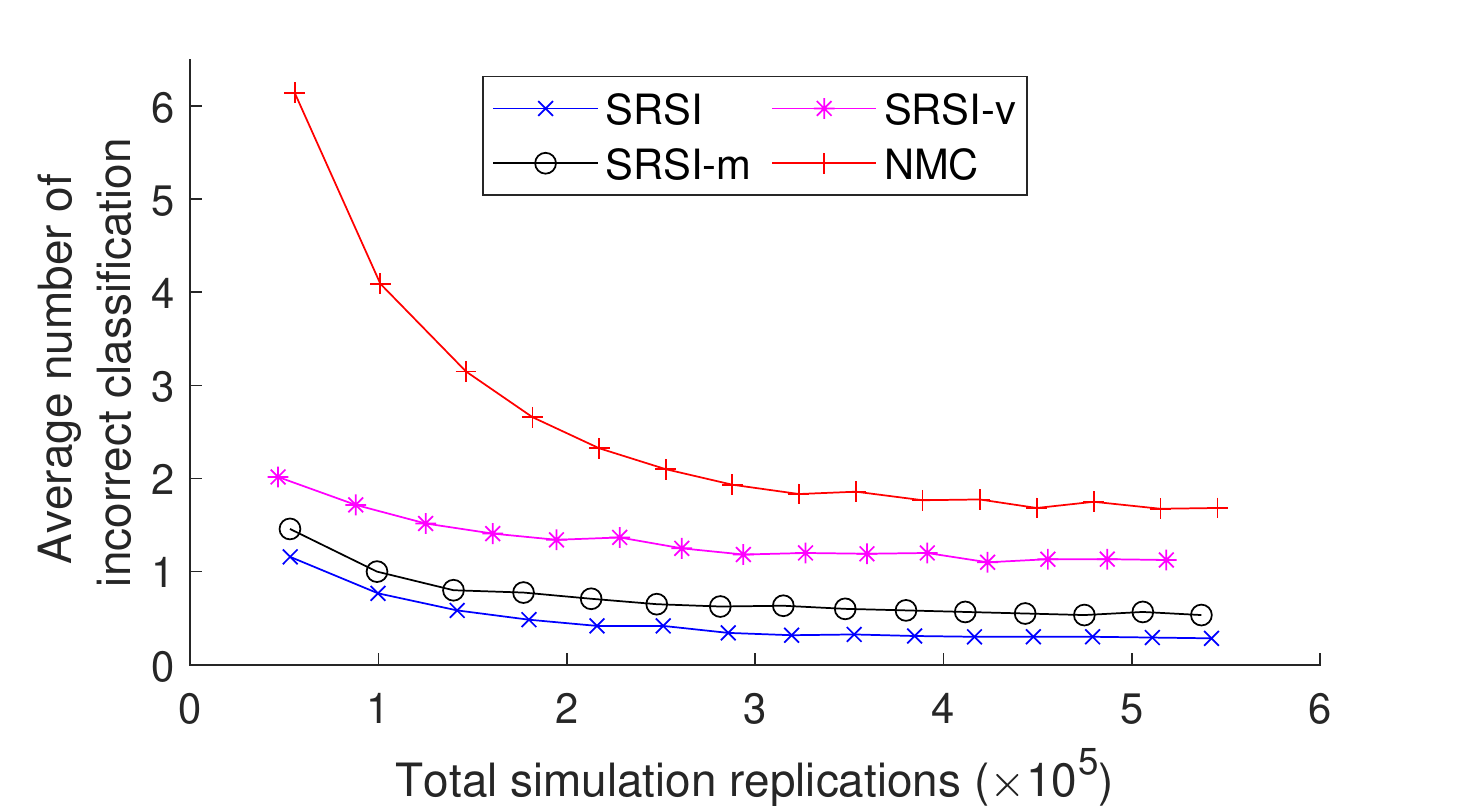}
\end{center}
\end{subfigure}
\caption{Empirical performances of the four procedures, SRSI, NMC, SRSI-m, and SRSI-v, from $120$ runs. Starting from the top left in clockwise order, each graph shows 1) the estimated probability that the estimated risk set from each procedure includes all elements of  $\tilde{S}^B_\alpha(\delta)$; 2) the estimated probability that the estimated risk set from each procedure equals to  $\tilde{S}^B_\alpha(\delta)$; 3) the average number of incorrectly classified solutions; 4) the sampling frequency of the solutions for a single run of SRSI, where $\bhx = 12$ and $\tilde{S}^B_\alpha(\delta) = \{15, 16, 17, 18, 19, 20\}$ for $\delta=1$ and $\alpha = 0.2$.}
\label{fig:SRSIcomparison}
\end{figure}

The results from all four procedures are shown in Figure~\ref{fig:SRSIcomparison}. The top left plot displays the probability that the estimated risk set from each procedure includes all elements of  $\tilde{S}^B_\alpha(\delta)$, i.e., probability of inclusion. SRSI has the highest probability of inclusion across all simulation budgets followed closely by SRSI-m, whereas SRSI-v whose distribution selection is only driven by GP prediction error performs worse than the other two. This shows importance of the distribution selection criterion for the performance of the sequential risk set inference procedure. Notice that the probability of inclusion has a downward trend for NMC. This is because for a small simulation budget, the standard error of $\bar{Y}(\bhx;\bP_b)$ is large and when it is overestimated, the NMC tends to include a lot of solutions in the risk set, which makes the estimated risk set much larger than $\tilde{S}^B_\alpha(\delta)$. 
The top right plot shows the probability that the estimated risk set from each procedure equals to  $\tilde{S}^B_\alpha(\delta)$, i.e., probability of identification. SRSI still dominates the other procedures. Notice there are bigger gaps among SRSI, SRSI-m, and SRSI-v than for the probability of inclusion.  NMC shows considerably poorer performance compared to the SRSI and its variations.
The bottom right plot shows the average number of incorrectly classified solutions. SRSI consistently has the smallest average across all simulation budgets;  SRSI has the average close to $1$ even for relatively simulation budget.  On the other hand, NMC has much larger average number of incorrect classifications than the rest especially when the simulation budget is smaller. 

To further demonstrate the performance of SRSI, we present the sampling frequency of all solutions within a single run of SRSI in the bottom left plot in Figure~\ref{fig:SRSIcomparison}. The $y$-axis represents the number of times each solution is sampled by the end of the $t=15{,}000$th iteration. For this run, $\bhx=12$ and $\tilde{S}^B_\alpha(\delta) = \{15, 16, 17, 18, 19, 20\}$. The most frequently sampled solution is $\bhx$, because all the rest of the solutions are compared with $\bhx$ and our sampling criterion ensures the GP has small prediction error at $\bhx$. Besides $\bhx,$ notice that more frequently sampled solutions are at the boundaries of $\tilde{S}^B_\alpha(\delta)$. These solutions are more difficult to classify than those at the interior of $\tilde{S}^B_\alpha(\delta)$ or clearly excluded solutions, thus our sampling criterion allocates more replications to them. 


\begin{table} [tb]
\caption{Comparison between $\widehat{S}_\alpha(\delta)$ and $\tilde{S}^B_\alpha(\delta)$ at different $\alpha$ levels from the same run as the bottom left plot of Figure~\ref{fig:SRSIcomparison}, where $\bhx = 12$. Distinct elements between two sets are in bold. }
\label{tab:alphaSets}
\resizebox{\textwidth}{!}{
\begin{tabular}{l|l|l}\hline
$\alpha$ & \multicolumn{1}{c|}{$\widehat{S}_\alpha(\delta)$} &  \multicolumn{1}{c}{$\tilde{S}^B_\alpha(\delta)$}\\ \hline
\multirow{3}{*}{$0.05$} & $\{6, 7, 8, 9, 10, 14, 15, 16, 17, 18, 19, 20, 21, 22,23,24,$ 
& $\{6, 7, 8, 9, 10, 14, 15, 16, 17, 18, 19, 20, 21, 22,23,24,$ \\
& $ 25, 26, 27, 28, 29, 30, 31, 32, 33, 34, 35, 36,37,38, 39,$ 
& $ 25, 26, 27, 28, 29, 30, 31, 32, 33, 34, 35, 36,37,38, 39, $ \\
& $ 40, 41, 42, 43, 44, 45, 46, 47, 48, \mathbf{49},\mathbf{50}\}$
& $ 40, 41, 42, 43, 44, 45, 46, 47, 48\}$\\ \hline
\multirow{2}{*}{$0.1$} & $\{7, 8, 9, 10, 14, 15, 16, 17, 18, 19, 20, 21, 22,23,24, 25,$
&$\{7, 8, 9, 10, 14, 15, 16, 17, 18, 19, 20, 21, 22,23,24,25, $\\
& $  26, 27, 28, 29, 30, 31, \mathbf{32}\}$
&$26, 27, 28, 29, 30, 31\}$\\\hline
$0.15$ & $\{9, 10, 15, 16, 17, 18, 19, 20, 21, 22, 23, 24, 25,\mathbf{26},\mathbf{27}\}$
& $\{\mathbf{8}, 9, 10, 15, 16, 17, 18, 19, 20, 21, 22, 23, 24, 25\}$\\\hline
$0.2$ & $\{15, 16, 17, 18, 19, 20\}$& $\{15, 16, 17, 18, 19, 20\}$\\\hline
$0.25$ & $\emptyset$& $\emptyset$\\
\end{tabular}}
\end{table}

\begin{table} [tb]
\caption{Comparison between $\widehat{S}_\alpha(\delta)$ and $\tilde{S}^B_\alpha(\delta)$ at different $\delta$ values from the same run as the bottom left plot of Figure~\ref{fig:SRSIcomparison}, where $\bhx = 12$. Distinct elements between two sets are in bold. }
\label{tab:deltaSets}
\resizebox{\textwidth}{!}{
\begin{tabular}{l|l|l}\hline
$\delta$ & \multicolumn{1}{c|}{$\widehat{S}_\alpha(\delta)$} &  \multicolumn{1}{c}{$\tilde{S}^B_\alpha(\delta)$}\\ \hline
$0$ & $\{8, 9, 10, 11, 13, 14, 15, 16, 17, 18, 19, 20, 21, 22, 23, \mathbf{24}\}$ 
& $\{8, 9, 10, 11, 13, 14, 15, 16, 17, 18, 19, 20, 21, 22, 23\}$  \\ \hline
$0.5$ & $\{8, 9,10, 11, 14, 15, 16, 17, 18, 19, 20, 21, \mathbf{22}\}$
&$\{8, 9,10, 11, 14, 15, 16, 17, 18, 19, 20, 21\}$\\ \hline
$1.0$ & $\{15, 16, 17, 18, 19, 20\}$& $\{15, 16, 17, 18, 19, 20\}$\\\hline
$1.5$ & $\emptyset$& $\emptyset$\\
\end{tabular}}
\end{table}

Although the experiments in Figure~\ref{fig:SRSIcomparison} were run with $\delta = 1$ and $\alpha = 0.2$, we may still estimate the risk sets at different $\alpha$ levels using the same GP posterior obtained from SRSI. In Table~\ref{tab:alphaSets}, we compare $\widehat{S}_\alpha(\delta)$ with the corresponding $\tilde{S}^B_\alpha(\delta)$ at different $\alpha$ levels for the same SRSI run displayed in the bottom left plot of Figure~\ref{fig:SRSIcomparison}, where distinct elements are in bold. For $\alpha \neq 0.2$, some solutions are misclassified, nevertheless, the estimated risk sets are quite close to $\tilde{S}^B_\alpha(\delta)$ at all levels of $\alpha$. 
The same observations can be made for the comparison between $\widehat{S}_\alpha(\delta)$ and $\tilde{S}^B_\alpha(\delta)$ at different $\delta$ values in Table~\ref{tab:deltaSets}.

\subsection{Ambulance dispatching center location problem}

In this section, we present a more realistic example modified from the DOvS case study in~\cite{wang2019}. Their goal is to decide a new location of an ambulance dispatching center in State College, Pennsylvania that minimizes the expected response time of the ambulances defined as  the time between receiving the emergency call and the patient's pick-up. 

In our simplified version, we consider a municipality consists of $|\cX| = 36$ neighborhoods on a grid in Figure~\ref{fig:freq.map}. Suppose the decision maker wants to select a neighborhood to place the dispatching center and is willing to ignore any difference in the expected response time within $\delta=1$ minute.  The center has $8$ ambulances in total. 
The color gradation of Figure~\ref{fig:freq.map} represents relative call frequency in each neighborhood calculated from $331$ emergency phone calls received by the center in the past; $40$ calls were collected from Neighborhood 30, whereas only one call was made from each of Neighborhoods 6, 11, and 15. The aggregate arrival process of emergency calls in the municipality is known as a Poisson process with rate $1$ call per hour, whereas the location distribution of the patients is unknown. When an emergency call is received, they dispatch the next available ambulance at the center to the patient's location and transport him/her to the center. If all ambulances are busy, the patient is put in the virtual queue until there is availability. All patients in the queue are first-come-first-served.  The ambulance's travel time between two neighborhoods is distributed as Erlang with scale $7.2$ minutes and phase equal to the Manhattan distance between the neighborhoods plus one. For instance, the travel time from Neighborhood 11 to 30 is distributed as Erlang($7.2, 5$). The travel time distribution is also assumed to be known. 

\begin{figure} [tb]
\begin{subfigure}{.5\textwidth}
\begin{center}
\includegraphics[scale=0.5]{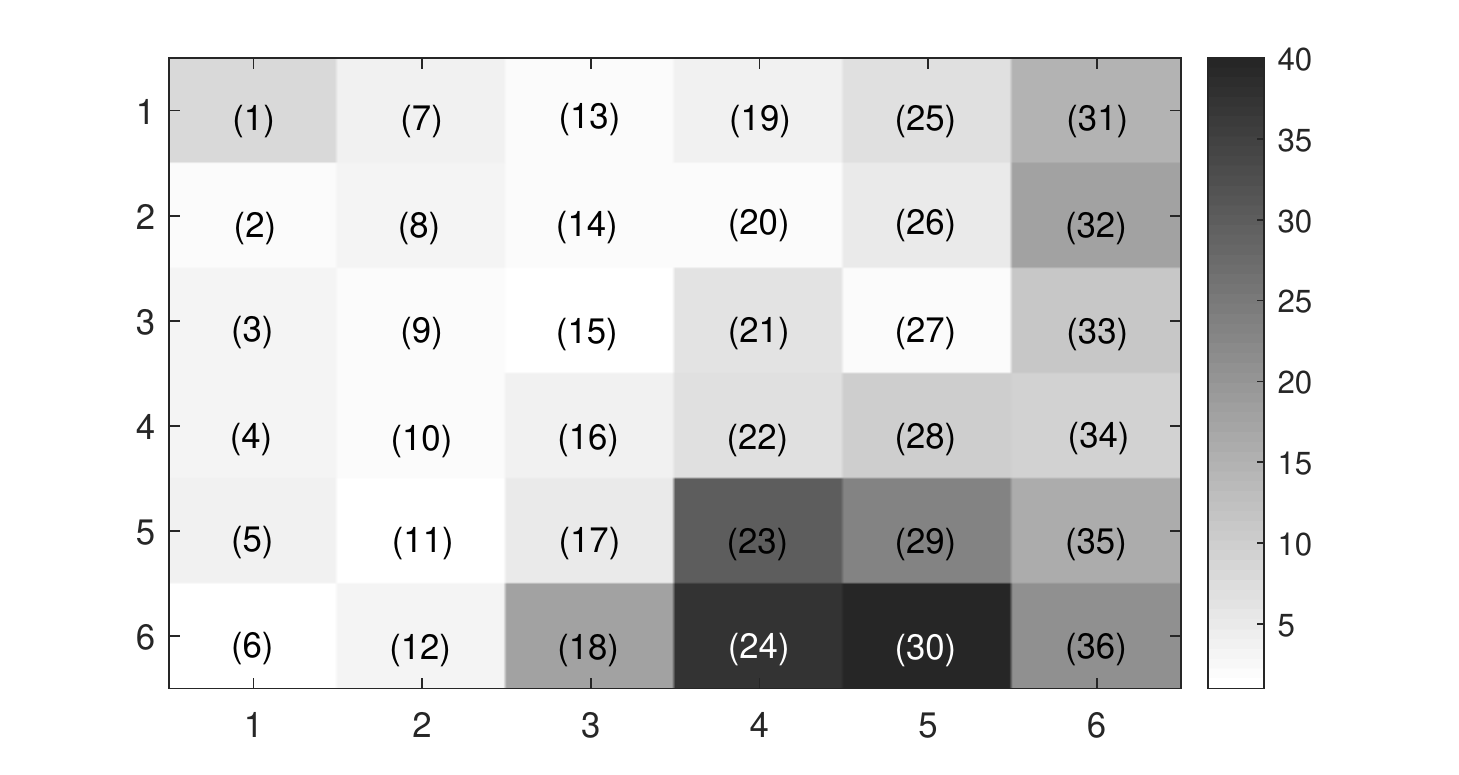}
\caption{The number of emergency phone calls received from $36$ neighborhoods}
\label{fig:freq.map}
\end{center}
\end{subfigure}%
\begin{subfigure}{.5\textwidth}
\begin{center}
\includegraphics[scale=0.5]{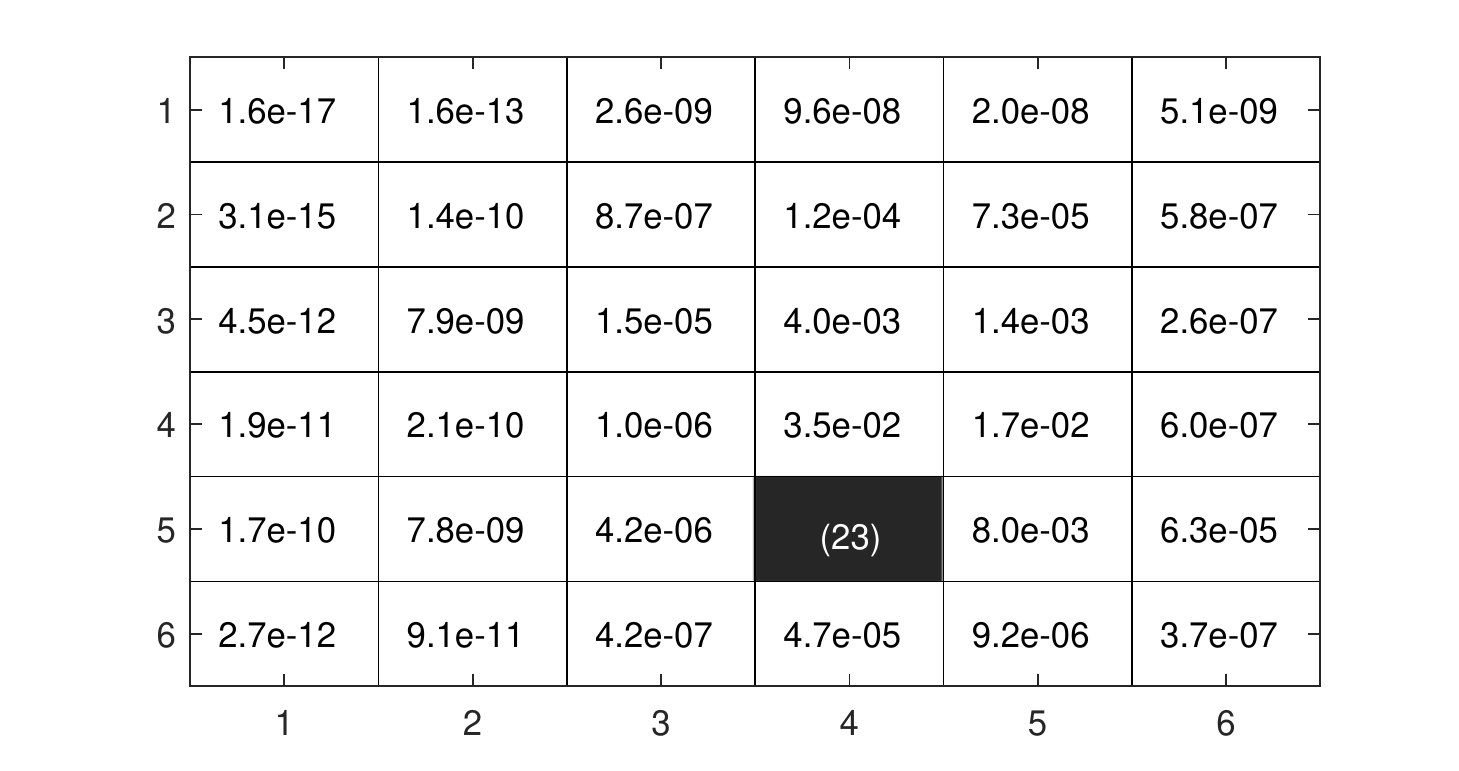}
\caption{The risk set of Neighborhood 23}
\label{fig:riskset23}
\end{center}
\end{subfigure}

\begin{subfigure}{.5\textwidth}
\begin{center}
\includegraphics[scale=0.5]{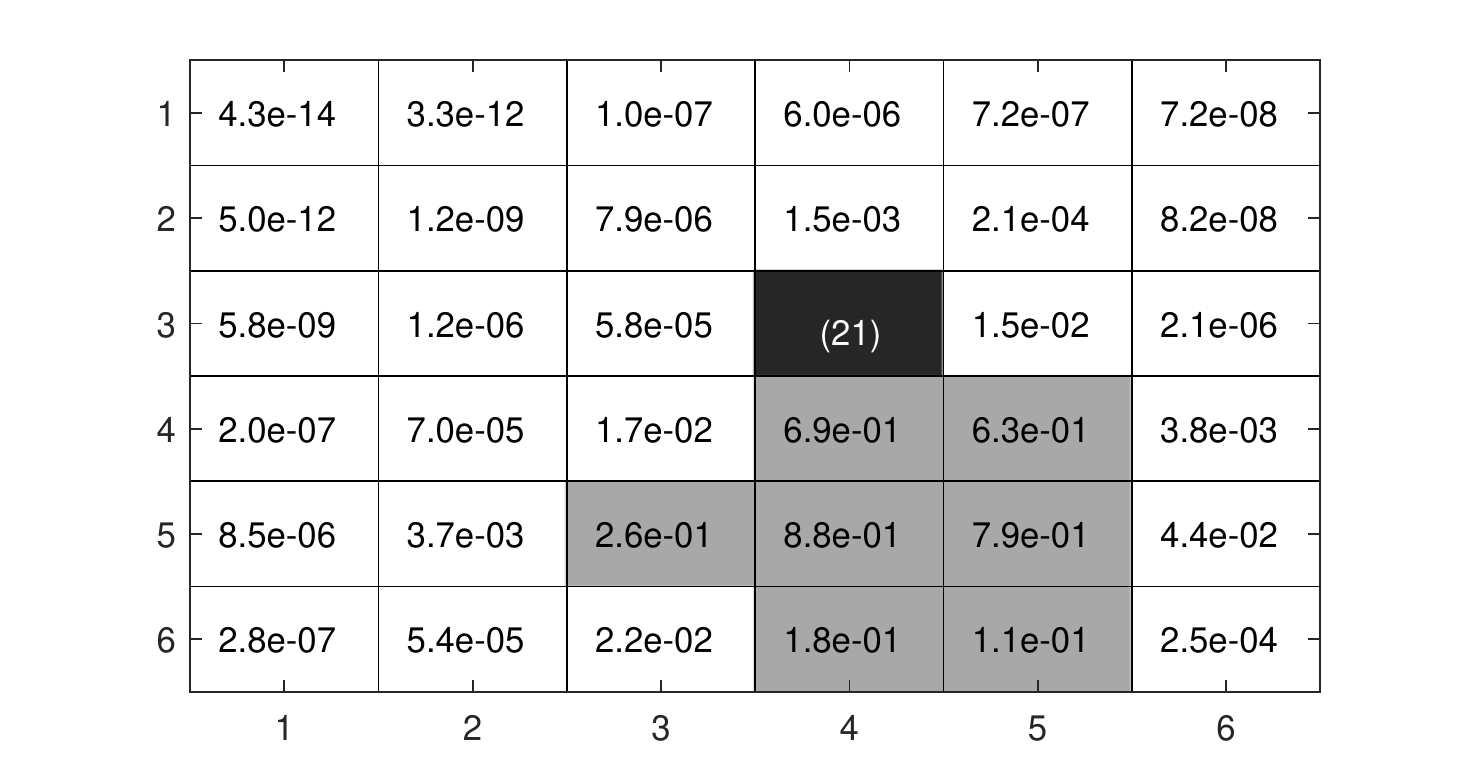}
\caption{The risk set of Neighborhood 21}
\label{fig:riskset21}
\end{center}
\end{subfigure}%
\begin{subfigure}{.5\textwidth}
\begin{center}
\includegraphics[scale=0.5]{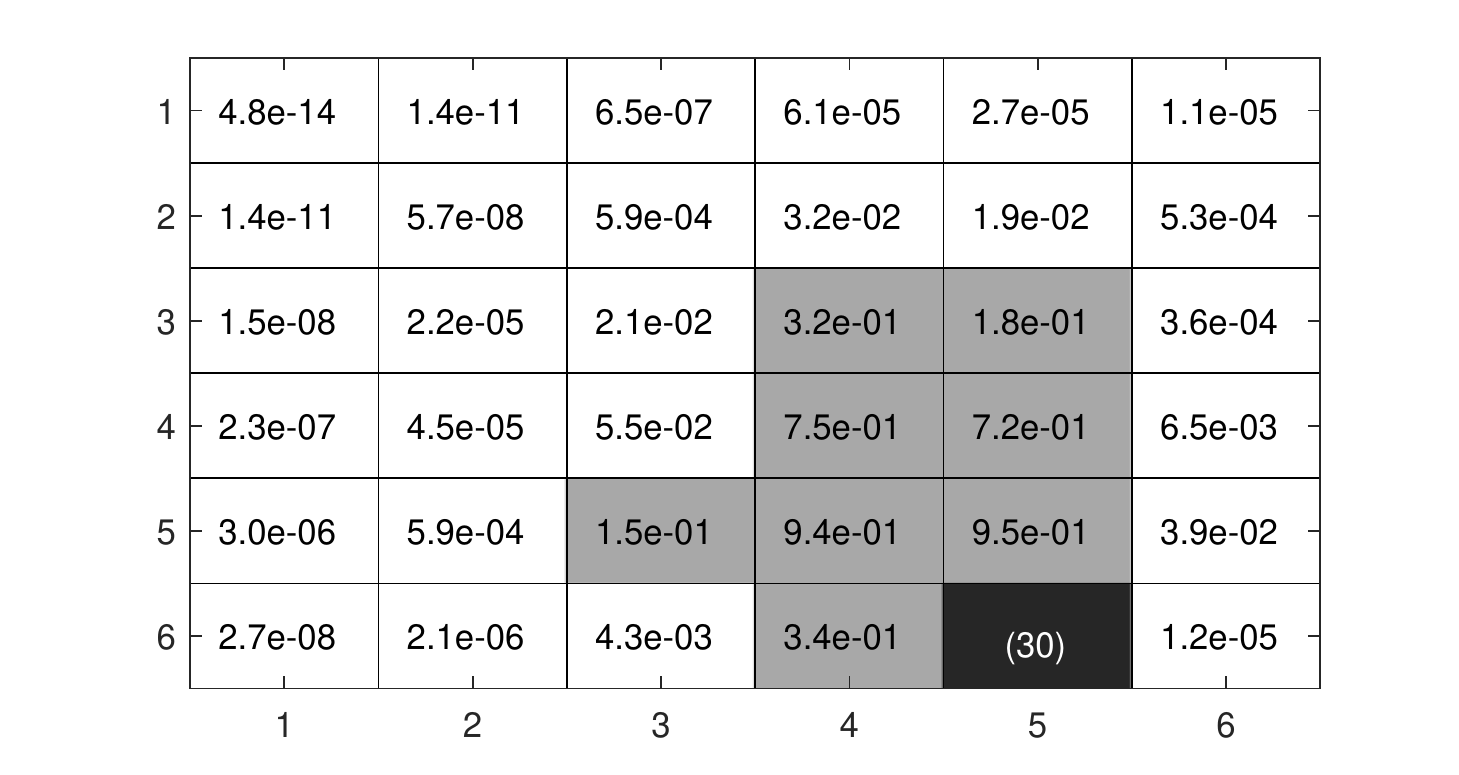}
\caption{The risk set of Neighborhood 30}
\label{fig:riskset30}
\end{center}
\end{subfigure}
\caption{Neighborhood maps of the ambulance dispatching center location problem. We set $\delta = 1$ minute and $\alpha = 0.1$. In~(b)--(d), $\bhx$ is in black and the solutions included in the estimated risk sets are in grey. The number in each region is the estimated probability that each solution is more than $1$-minute shorter than $\bhx$ given the posterior on the patient location distribution. }
\end{figure}

To summarize, the only source of input uncertainty in this problem is the location distribution of the patients, which we  model nonparametrically using the Dirichlet prior described  and we set $\kappa_j = 1$ for  $j=1,2,\ldots,36$. Thus, the MAP of the Dirichlet posterior distribution, $P_\mathrm{MAP}$, is the same as the empirical distribution of the call frequency data in Figure~\ref{fig:freq.map}. 

We built a discrete-event simulator using Python to estimate the steady-state mean of the response time for each solution. For each simulated emergency call, its location is randomly generated from the probability simplex that represents the location distribution. To eliminate the initial bias, $1{,}000$-hour warm-up period is used for all solutions. Each replication averages the response time of patients picked up by the ambulances during $50$ hours after the warm up. 

Suppose the decision maker ran a R\&S procedure to find Neighborhood 23 to be the optimum conditional on $P_\mathrm{MAP}$---we confirmed this by replicating all solutions $10{,}000$ times to estimate the conditional expected response time with less than $4\%$ relative errors---and wants to further investigate input model risk of implementing this decision. She chose $\alpha = 0.1$ (risk averse). 

We ran SRSI with $B = 150, n_0 = 108,$ and $r=2$. We continued to run $R_t=2$ replications for each selected solution-distribution pair and stopped at the $t=100$th iteration. Figure~\ref{fig:riskset23} shows the estimated probability that each solution's expected response time is more than $1$-minute shorter than that of Neighborhood 23 given the posterior on the patient location distribution. Notice that all solutions' probabilities are less than $0.1$ leaving the estimated risk set empty. This indicates Neighborhood 23 is a robust solution to input model risk. 

Along with the conditional optimum, we also ran the procedure to estimate the risk sets of Neighborhoods 21 (geographic center) and 30 (population center) as presented in Figures~\ref{fig:riskset21} and \ref{fig:riskset30}, respectively. In the former, notice that the solutions in the risk set are closer to the high-frequency neighborhoods in Figure~\ref{fig:freq.map} than Neighborhood 21 is. These locations tend to perform better than Neighborhood 21 as they can serve a large portion of the emergency patients with shorter response time. On the other hand, Neighborhood 36 is excluded from the risk set despite its proximity to the high-frequency neighborhoods. This is because of its distance from the upper left quadrant of the municipality; although they receive emergency calls less frequently from the quadrant, placing the dispatching center at Neighborhood~36 substantially increases the expected response time of the patients in the quadrant. Given a probability simplex sampled from the posterior Dirichlet that has relatively high probabilities on the upper left quadrant, Neighborhood~36 performs significantly  worse than Neighborhood 21 causing it to be excluded from the risk set. 
Similar observations can be made for the risk set of Neighborhood 30 in Figure~\ref{fig:riskset30}.


\section*{Acknowledgment}
This work is supported by the NSF Grant No.\ 1854659.

\onehalfspacing
\bibliographystyle{chicago}
\bibliography{SequentialRiskSet}

\doublespacing

\end{document}